\documentclass[10pt,aps,pra,twocolumn,floatfix,nofootinbib,superscriptaddress,longbibliography]{revtex4-2}
\usepackage{placeins}

\usepackage{bm,graphicx,mathrsfs,amsmath,amssymb,mathtools,,makecell,bbm, amsthm, dsfont,color,times,txfonts,nicefrac,framed,enumitem,tikz,physics,wrapfig,amsfonts,tcolorbox,lipsum,nicematrix,multirow}
\usepackage[ruled,vlined, linesnumbered]{algorithm2e}

\usepackage[colorlinks=true,linkcolor=refcolor,citecolor=refcolor]{hyperref}

\hypersetup{urlcolor=refcolor}
\definecolor{equationcolor}{RGB}{108,153,224}
\definecolor{refcolor}{RGB}{135,147,216}
\definecolor{changescolor}{RGB}{188, 104, 104}

\makeatletter
\def\blfootnote{\gdef\@thefnmark{}\@footnotetext}
\makeatother

\newcommand{\ms}[1]{\textsf{#1}}

\newcommand{\iden}{\mathbbm{1}}

\newcommand{\E}[1]{\mathcal{E}}

\def\E{ {\ms E} }

\newtheorem{thm}{Theorem}

\newtheorem{lem}[thm]{Lemma}
\newtheorem{prop}[thm]{Proposition}
\newtheorem{cor}[thm]{Corollary}

\newtheorem{defn}{Definition}

\begin{document}
\title{Predicting Magic from Very Few Measurements}

\author{J. M. Varela}
    \affiliation{Physics Department, Federal University of Rio Grande do Norte, Natal, 59072-970, Rio Grande do Norte, Brazil}
    \affiliation{International Institute of Physics, Federal University of Rio Grande do Norte, 59078-970, Natal, RN, Brazil}
\author{L. L. Keller}
    \affiliation{International Institute of Physics, Federal University of Rio Grande do Norte, 59078-970, Natal, RN, Brazil}
\author{A. de Oliveira Junior}
	\affiliation{Center for Macroscopic Quantum States bigQ, Department of Physics,
Technical University of Denmark, Fysikvej 307, 2800 Kgs. Lyngby, Denmark}
\author{D. A. Moreira}
    \affiliation{School of Science and Technology, Federal University of Rio Grande do Norte, Natal, Brazil}
\author{R. Chaves}
    \affiliation{International Institute of Physics, Federal University of Rio Grande do Norte, 59078-970, Natal, RN, Brazil}
    \affiliation{School of Science and Technology, Federal University of Rio Grande do Norte, Natal, Brazil} 
\author{R. A. Macêdo}
    \affiliation{Physics Department, Federal University of Rio Grande do Norte, Natal, 59072-970, Rio Grande do Norte, Brazil}
    \affiliation{International Institute of Physics, Federal University of Rio Grande do Norte, 59078-970, Natal, RN, Brazil}
\date{\today}

\begin{abstract}
The nonstabilizerness of quantum states is a necessary resource for universal quantum computation, yet its characterization is notoriously demanding. Quantifying nonstabilizerness typically requires an exponential number of measurements and a doubly exponential classical post-processing cost to evaluate its standard monotones. In this work, we show that nonstabilizerness is, to a large extent, in the eyes of the beholder: it can be witnessed and quantified using \emph{any} set of $m$ $n$-qubit Pauli measurements, provided the set contains anti-commuting pairs. We introduce a general framework that projects the stabilizer polytope onto the subspace defined by these observables and provide an algorithm that estimates magic from Pauli expectation values with runtime exponential in the number of measurements $m$ and polynomial in the number of qubits $n$. By relating the problem to a stabilizer-restricted variant of the quantum marginal problem, we also prove that deciding membership in the corresponding reduced stabilizer polytope is NP-hard. In particular, unless $\mathrm{P} = \mathrm{NP}$, no algorithm polynomial in $m$ can solve the problem in full generality, thus establishing fundamental complexity-theoretic limitations. Finally, we employ our framework to compute nonstabilizerness in different Hamiltonian ground states, demonstrating the practical performance of our method in regimes beyond the reach of existing techniques.
\end{abstract}

\maketitle

\section{Introduction}
A central problem in quantum information science is to identify and quantify the resources that allow quantum systems to outperform their classical counterparts~\cite{nielsen2010quantum, wilde2013quantum}. In the context of quantum computation, this distinction is sharply captured by the theory of stabilizer states and operations~\cite{gottesman1997stabilizer}. While stabilizer states may exhibit large amounts of entanglement, the Gottesman–Knill theorem guarantees their efficient classical simulability~\cite{gottesman1998heisenberg, aaronson2004improved}, so they are insufficient for universal quantum computation. The essential ingredient behind quantum advantage is then called nonstabilizerness, or magic, which has therefore emerged as a fundamental resource for fault-tolerant quantum computation~\cite{bravyi2005universal} as well as in other quantum information processing tasks~\cite{gu2024pseudomagic,cao2025gravitational,zamora2025semi,junior2025athermality,loio2025quantum}.

Despite its importance, the characterization of nonstabilizerness remains well known to be difficult. At the heart of the problem lies the complex geometry of the stabilizer polytope, the convex hull of all stabilizer states~\cite{Veitch_2014}. Determining whether a given quantum state lies inside or outside this polytope or quantifying its distance from it poses two major bottlenecks. The first bottleneck is experimental. Most existing measures and witnesses of nonstabilizerness rely on complete state tomography, which requires estimating $4^n$ expectation values for an $n$-qubit system and is therefore infeasible beyond small system sizes. The second bottleneck is computational. Even assuming full tomographic data were available, deciding whether a state lies inside the stabilizer polytope is itself an intractable task: the number of its vertices grows exponentially with the number of qubits. As a consequence, existing approaches must rely on strong symmetry assumptions, highly structured families of states, or severe dimensional restrictions~\cite{heinrich2019robustness, timsina2025robustness}. This exponential overhead contrasts with the scale of current quantum platforms, which already operate at--and in some cases beyond--hundreds of qubits~\cite{Arute2019,Zhong2020,Ebadi2021,Wu2021}.

Together, these obstacles severely limit our ability to detect and quantify nonstabilizerness in large quantum systems, and raise a fundamental question: Can one quantify the degree of nonstabilizerness of a quantum state using only a limited number of measurements, while simultaneously avoiding computationally intractable optimization over the full stabilizer polytope? In this work, we answer this question in the affirmative. Inspired by the philosophy of classical shadows~\cite{huang2020predicting}, we show that it is possible to capture essential information about nonstabilizerness by probing only a low-dimensional subspace of the full Hilbert space. Rather than reconstructing the entire quantum state, we project the stabilizer polytope onto this subspace and study the resulting reduced marginal polytope. Remarkably, this projection preserves enough geometric structure to allow for faithful detection and quantification of nonstabilizerness, while drastically reducing both measurement and computational complexity.

We develop a near-optimal algorithm to construct such projected stabilizer polytopes using a number of measurements that scales polynomially with the system size. This approach simultaneously addresses the two aforementioned bottlenecks: it replaces full tomography with a shadow-like measurement scheme, and it reduces the complexity of the convex optimization problem by working with a significantly smaller polytope. It is characterized by the commutativity graph of Pauli measurements, similar to the ideas developed in~\cite{gokhale2019minimizing,chapman2020characterization,mann2025graph,xu2025simultaneous}. We illustrate the general framework with several explicit examples, demonstrating that our method reliably quantifies nonstabilizerness in regimes far beyond the reach of existing techniques.

The paper is organized as follows. In Sec.~\ref{sec:stab}, we review the essentials of stabilizer resource theory and the robustness of magic (RoM), the central monotone of nonstabilizerness employed throughout this work. In Sec.~\ref{sec:reduxstab}, we introduce the reduced stabilizer polytope and present a near-optimal algorithm for its construction, together with the associated nonstabilizerness monotone. We also provide a detailed analysis of several reduced polytopes of particular interest. In Sec.~\ref{sec:stabmargprob}, we address the computational complexity of the problem, proving that even for reduced stabilizer polytopes there exists no polynomial-time algorithm to decide whether a given set of correlations lies inside the polytope. In Sec.~\ref{sec:hamilton}, we apply our general framework to selected cases of interest, computing the nonstabilizerness of ground states of several Hamiltonians using only the local observables required to determine their energy. Finally, in Sec.~\ref{sec:discuss}, we summarize our findings and highlight directions for future research.

\section{Stabilizer resource theory and monotones}
\label{sec:stab}

In this section, we review the basic ingredients underlying the definition of stabilizer states and summarize the relevant resource monotones. We consider systems of $n$ qubits, described by the Hilbert space $\mathcal H = (\mathbb{C}{^2)^{\otimes n}}$. We denote by $\mathcal P_n$ the $n$-qubit Pauli group, defined as the group generated by tensor products of single-qubit Pauli operators, including overall phases. When convenient, we will refer to $\tilde{\mathcal P}_n$ as the Pauli group modulo phases, i.e., the Pauli group without signs.

A stabilizer group $\mathcal S \subset \mathcal P_n$ is an Abelian subgroup of the Pauli group such that $-\iden \notin \mathcal S$. If $\mathrm{rank}(\mathcal S)=n$, the group $\mathcal S$ uniquely defines a pure stabilizer state $|\psi^{\mathcal S}\rangle$ as the joint $+1$ eigenstate of all its elements, that is,
\begin{equation}
    s\ket{\psi^\mathcal{S}} = +\ket{\psi^\mathcal{S}} \:\: \forall s\in\mathcal{S}.
\end{equation}
Because pure stabilizer states are in one-to-one correspondence with stabilizer groups of maximal rank, we denote both sets by $\mathrm{stab}_n$.

The set of (possibly mixed) stabilizer states is defined as the convex hull of pure stabilizer states,
\begin{align}
\mathrm{STAB}_n &\equiv \mathrm{conv}(\mathrm{stab}_n) \nonumber\\
&= \left\{\sum_{\mathcal S \in \mathrm{stab}_n} p_{\mathcal S}\,\psi^{\mathcal S} \;\middle|\; \vec p \in \Delta_{|\mathrm{stab}_n|}\right\},
\label{eq:stabpolytope}
\end{align}
where $\Delta_{|\mathrm{stab}_n|} = \{\vec p \in \mathbb R^{|\mathrm{stab}_n|} \mid p_{\mathcal S} \geq 0,\; \sum_{\mathcal S} p_{\mathcal S}=1\}$ denotes the probability simplex and $\psi^{\mathcal S} = |\psi^{\mathcal S}\rangle\!\langle\psi^{\mathcal S}|$. The set $\mathrm{STAB}_n$ thus forms a polytope in the space of quantum states; the above characterization as the convex hull of finitely many extremal points is known as its $V$-representation~\cite{ziegler1995lectures}.

There exists a natural class of unitaries that preserve the stabilizer polytope and form a group, known as the Clifford group. These unitaries are defined as those that conjugate Pauli operators into Pauli operators,
\begin{equation}
    \mathcal C\ell_n \!\equiv\! \left\{ U \in U(\mathcal H) \middle| \forall P \in \mathcal P_n,\; \exists P^\prime \in \mathcal P_n \text{ s.t. } U P U^\dagger \propto P^\prime \right\} .
\end{equation}
A quantum channel $\mathcal E$ that preserves $\mathrm{STAB}_n$—that is, $\mathcal E(\psi) \in \mathrm{STAB}_n$ for all $\psi \in \mathrm{STAB}_n$—is called a \emph{stabilizer operation}. Any such channel can be decomposed into a sequence of Clifford unitaries, computational-basis measurements, and classical post-processing. Together, the stabilizer polytope and the set of stabilizer operations define the stabilizer resource theory~\cite{Veitch_2014}, in which $\mathrm{STAB}_n$ constitutes the set of free states and stabilizer operations are the free operations.

A generic quantum state $\psi$ is said to be \emph{nonstabilizer} if $\psi \notin \mathrm{STAB}_n$. Nonstabilizerness constitutes the resource of the stabilizer resource theory, in direct analogy with entanglement, coherence, nonlocality, and other quantum resources~\cite{chitambar2019quantum}. To quantify this resource, one introduces a \emph{stabilizer monotone}, namely a real-valued function $f$ on quantum states satisfying the following properties: \emph{(i)} all stabilizer states attain the same value $f_0$, and \emph{(ii)} $f(\mathcal E(\psi)) \leq f(\psi)$ for all stabilizer operations $\mathcal E$. We say that $f$ is \emph{faithful} if $f(\psi) \geq  f_0$ for all states, with equality if and only if $\psi \in \mathrm{STAB}_n$.

Various stabilizer monotones have been proposed in the literature. These include measures based on the trace distance~\cite{junior2025geometric} and the relative entropy~\cite{Veitch_2014}, which are generic constructions for convex quantum resource theories. Such monotones can typically be expressed as
\begin{equation}
    f(\psi) = \min_{\sigma \in \mathrm{STAB}_n} D(\psi, \sigma)\;,
\end{equation}
where $D(\cdot,\cdot)$ denotes a convex distance measure (or, in the case of concave entropic measures, an analogous expression in which the optimization is a maximization over stabilizer states). For simpler, non-optimizing entropic quantifiers, one must restrict to the pure-state stabilizer resource theory~\cite{leone2022stabilizer}, as is also common in the resource theory of entanglement~\cite{nielsen2010quantum}.

By Eq.~\eqref{eq:stabpolytope}, optimization over stabilizer states requires $|\mathrm{stab}_n| =2^{\Theta(n^2)}$~\cite{aaronson2004improved} data even to specify $\sigma \in \mathrm{STAB}_n$ in the $V$-representation. Moreover, computing the monotone itself typically requires tomographically complete information about the state $\psi$—as is the case for both the relative entropy~\cite{Veitch_2014} and the trace distance~\cite{junior2025geometric}, for example—entailing an experimental overhead scaling as $O(4^n)$.

In this work, our monotone of interest will be the \emph{Robustness of Magic} (RoM), defined for a quantum state $\psi$ as~\cite{howard2017application}
\begin{equation}
    \mathrm{RoM}(\psi) = \min_{\vec x} \left\{ \| \vec x\|_1 \;\middle|\; \psi=\sum_{\mathcal S \in \mathrm{stab}_n}x_\mathcal S \psi^\mathcal S \right\}\;, 
    \label{eq:RoM}
\end{equation}
where the vector $\vec x \in \mathbb R^{|\mathrm{stab}_n|}$ specifies a \emph{stabilizer pseudo-mixture} of $\psi$, and $\| \vec x\|_1 = \sum_{\mathcal S \in \mathrm{stab}_n}|x_\mathcal S|$ is referred to as its \emph{negativity}. Since $\tr \psi=1$, the coefficients necessarily satisfy $\sum_{\mathcal S \in \mathrm{stab}_n}x_\mathcal S =1$, implying the lower bound $\mathrm{RoM}(\psi)\geq 1$. The Robustness of Magic is a faithful stabilizer monotone~\cite{howard2017application}, that is, the minimum value is attained if and only if $\psi$ is a stabilizer state, and the minimization ensures monotonicity under all stabilizer operations.

In the original work, it was shown that the evaluation of the RoM can be formulated as a linear program (LP),
\begin{align}
     \mathrm{RoM}(\psi) =&\min_{\vec x} \| \vec x\|_1\;, \nonumber\\     &\:\:\mathrm{s.t.}\;\;  \mathbf A \vec x = \vec b \;,
\end{align}
where the elements of the rectangular matrix $\mathbf A$ are given by $A_{P,\mathcal S} = \tr(P \psi^\mathcal S)$ and the vector $\vec b$ is defined as $b_P = \tr(P \psi)$. Despite this reformulation, the problem remains plagued by the same fundamental limitations discussed above: in its current form, it requires an exponential amount of experimental data and computational resources to be evaluated.

Beyond its role as a resource monotone, the RoM also quantifies the classical simulability of quantum states via quasiprobability methods~\cite{pashayan2015estimating}. In particular, it was shown in~\cite{howard2017application} that expectation values of the form $\tr(P \mathcal E(\psi))$ can be estimated within additive error $\delta$ and failure probability $\epsilon$ in $\mathrm{poly}(n)$ time, provided that
\begin{equation}
    N \geq \frac{2}{\delta^2}\mathrm{RoM}(\psi)^2 \ln \left(\frac{2}{\epsilon}\right),
    \label{eq:RoMsamplecomp}
\end{equation}
independent samples are drawn from the probability distribution defined by $p_\mathcal S = |x_\mathcal S|/\|\vec x\|_1$.

In the next section, we turn to the projection of the stabilizer resource theory in settings where tomographic completeness is not available. In particular, we discuss how to construct suitable projections of the stabilizer polytope and define the corresponding RoM measure.

\section{Reduced stabilizer polytope and robustness}
\label{sec:reduxstab}

In realistic experimental settings, access to tomographically complete data is rarely available. Instead, measurements are typically constrained by the physical architecture of the platform, calibration requirements, and noise considerations, resulting in access to only a limited set of observables. These observables are often fixed in advance, for instance, by native measurement bases, locality constraints, or readout fidelities, and may scale at most polynomially with the system size. As a consequence, any characterization of nonstabilizerness that fundamentally relies on full state tomography is incompatible with the operational regime of contemporary quantum devices. This motivates the following shift in perspective: rather than attempting to reconstruct the full quantum state, we ask what can be inferred about nonstabilizerness from partial information alone. In particular, given expectation values of a restricted set of experimentally accessible observables, can one still certify and quantify nonstabilizerness in a meaningful and resource-efficient way? Addressing this question requires rethinking the geometry of the stabilizer polytope under projection onto lower-dimensional subspaces, which is the approach we will develop in the following. The key idea behind it is to construct a \emph{projected stabilizer polytope}.

We begin by assuming that the experimenter has access to only $m$ Pauli operators, collected into the set
\begin{equation}
    \mathcal M = \{P_1, P_2, \cdots, P_m\}\;,
\end{equation}
where we always take $m \leq \mathrm{poly}(n)$. Given a quantum state $\psi$, we may represent it by its Pauli expectation values,
\begin{equation}
    \vec{\psi}\equiv (\langle P \rangle_\psi )_{P \in \tilde{\mathcal P}_n} \in [-1,1]^{4^n}\;.
\end{equation}
In this representation, the stabilizer polytope can be written as
\[
\mathrm{conv}\!\left(\vec{\psi}^{\mathcal S}\right)_{\mathcal S \in \mathrm{stab}_n} \subset [-1,1]^{4^n}\;.
\]
Denoting by $\Pi_\mathcal M$ the projection that maps a vector $\vec v \in [-1,1]^{4^n}$ onto the components corresponding to the observables in $\mathcal M$, we define the \emph{projected stabilizer polytope} as the image of $\mathrm{STAB}_n$ under this projection,
\begin{equation}
    \mathrm{STAB}_n(\mathcal M)\equiv \Pi_\mathcal M[\mathrm{STAB}_n]\;.
\end{equation}
Clearly, if $\mathcal M = \tilde{\mathcal P}_n$, one recovers the full stabilizer polytope, i.e., $\mathrm{STAB}_n(\tilde{\mathcal P}_n)=\mathrm{STAB}_n$.

A noteworthy property of this construction is that the resulting projected polytope depends only weakly on the total number of qubits. In particular, we have the following result.
\begin{lem}
    Let $\{P_1, P_2, \cdots, P_m\}$ be a set of $n$-qubit Pauli operators, and let $n^\prime \geq n$. Then
    \begin{equation}
        \mathrm{STAB}_n(\{P_i\}_{i=1}^{m}) = \mathrm{STAB}_{n^\prime}\big(\{P_i \otimes \mathbb I_{n^\prime-n}\}_{i=1}^{m}\big)\;,
    \end{equation}
    where the identity acts on the remaining $n^\prime-n$ qubits.
    \label{lem:weakdepstab}
\end{lem}
The proof is given in Appendix~\hyperref[proof:weakdepstab]{A-1}. This result suggests that the projected stabilizer polytope depends only on the algebraic relations among the Pauli operators in $\mathcal M$, rather than on their explicit tensor-product embedding into the full Hilbert space. We will make this intuition precise in the next subsection.

Since specifying the measurement set $\mathcal M$ already fixes the number of qubits required to realize the corresponding observables, we will henceforth suppress the explicit dependence on $n$ and write $\mathrm{STAB}(\mathcal M) \equiv \mathrm{STAB}_n(\mathcal M)$. We further denote by $\mathrm{stab}(\mathcal M) \subseteq \mathrm{STAB}(\mathcal M)$ a generating set for the $V$-representation of the projected polytope, namely
\begin{equation}
\!\!\!\mathrm{STAB}(\mathcal M) \!=\! \mathrm{conv}[\mathrm{stab}(\mathcal M)] \!\equiv\! \mathrm{conv}(\vec{v}_1, \vec{v}_2, \cdots, \vec{v}_{|\mathrm{stab}(\mathcal M)|})\;.
\end{equation}
Because the projection can map distinct stabilizer states to the same point—or render some points redundant via convex combinations—one generally has
$\mathrm{stab}(\mathcal M) \subsetneq \Pi_{\mathcal{M}}\mathrm{stab}_n$. For instance, it may occur that
\[
\Pi_\mathcal M \vec{\psi}^{\mathcal S}
=
\sum_{\mathcal S^\prime \neq \mathcal S} p_{\mathcal S^\prime}\,
\Pi_\mathcal M \vec{\psi}^{\mathcal S^\prime}\;.
\]
We postpone the explicit construction of the $V$-representation $\mathrm{stab}(\mathcal M)$ to the next subsection and treat it as an input for the remainder of the present discussion.

Observing correlations that lie outside the projected stabilizer polytope provides a bona fide witness of nonstabilizerness. Indeed, for a given state $\psi$,
\begin{equation}
    (\langle P_1\rangle_\psi, \dots, \langle P_m \rangle_\psi) \notin \mathrm{STAB}(\mathcal M)
    \;\Rightarrow\;
    \psi \notin \mathrm{STAB}_n\;.
    \label{eq:reducedstabwitness}
\end{equation}
While this criterion suffices for witnessing nonstabilizerness, it does not provide a quantitative measure. To go beyond mere detection, we seek an analogue of a stabilizer monotone that can be evaluated solely from the expectation values of the observables in $\mathcal M$.

To this end, we introduce a natural generalization of the Robustness of Magic in Eq.~\eqref{eq:RoM}, which we call the \emph{reduced Robustness of Magic},
\begin{align}
    \mathrm{RoM}_\mathcal M(\psi)
    =&\min_{\vec x \in \mathbb R^{|\mathrm{stab}(\mathcal M)|}} \| \vec x\|_1\;,
    \label{eq:reducedRoM}\\
    &\hspace{0.5cm}\mathrm{s.t.}\;\; \Pi_\mathcal M\vec{\psi}
    =
    \sum_{j=1}^{|\mathrm{stab}(\mathcal M)|}x_j \vec{v}_j\;, \nonumber\\
    &\phantom{\mathrm{s.t.}}\;\;
    \sum_{j=1}^{|\mathrm{stab}(\mathcal M)|}x_j = 1 \;.
\end{align}
As in the full RoM, this definition leads to a linear program, but now expressed exclusively in terms of the measurement data associated with $\mathcal M$.

The reduced Robustness of Magic satisfies the following properties.
\begin{lem}
    Given any set of $n$-qubit Pauli operators $\mathcal M = \{P_1, P_2, \cdots, P_m\}$ and a quantum state $\psi$, the projected robustness of magic $\mathrm{RoM}_\mathcal M$ obeys:
    \begin{enumerate}
        \item $\mathrm{RoM}_\mathcal M(\psi) \geq 1$, with equality if and only if there exists a stabilizer state $\phi \in \mathrm{STAB}_n$ such that $\langle P_i \rangle_\psi = \langle P_i \rangle_\phi$ for all $i \in \{1,2,\cdots,m\}$;
        \item $\mathrm{RoM}_\mathcal M(\mathcal E(\psi)) \leq \mathrm{RoM}_\mathcal M(\psi)$ for all trace-preserving stabilizer operations $\mathcal E$;
        \item $\mathrm{RoM}_\mathcal M$ is submultiplicative and convex, i.e.,
        $\mathrm{RoM}_\mathcal M(\psi \otimes \rho) \leq \mathrm{RoM}_\mathcal M(\psi)\,\mathrm{RoM}_\mathcal M(\rho)$ and
        $\mathrm{RoM}_\mathcal M\!\left(\sum_j p_j \psi_j\right) \leq \sum_j |p_j|\,\mathrm{RoM}_\mathcal M(\psi_j)$;
        \item $\mathrm{RoM}(\psi) \geq \mathrm{RoM}_\mathcal M(\psi)$.
    \end{enumerate}
\label{lem:reducedROMproperties}
\end{lem}
The proof is provided in Appendix~\ref{proof:reducedromproperties}. Properties~2 and~3 show that the reduced Robustness of Magic is indeed a valid stabilizer monotone, inheriting many of the structural features of the full Robustness of Magic~\cite{howard2017application}. The most significant insights, however, arise from properties~1 and~4.

Property~1 highlights that $\mathrm{RoM}_\mathcal M$ is generally \emph{not} faithful: the condition $\mathrm{RoM}_\mathcal M(\psi)=1$ guarantees only that there exists a stabilizer state reproducing the statistics of $\psi$ on the restricted measurement set $\mathcal M$. Consequently, depending on the choice of observables, nonstabilizer states may remain indistinguishable from stabilizer states at the level of accessible Pauli expectations. A precise characterization of this phenomenon will be developed in Sec.~\ref{sec:stabmargprob}.

Property~4 ensures that the measurement-limited robustness nonetheless provides a meaningful lower bound on the nonstabilizerness of the full state. In particular, it implies a corresponding lower bound on the classical simulation cost: by measuring only $m \leq \mathrm{poly}(n)$ Pauli operators, one can already infer a sample-complexity bound of the form
$N \geq f(\epsilon,\delta)\,\mathrm{RoM}_\mathcal M(\psi)^2$
for quasiprobability-based simulation, directly from Eq.~\eqref{eq:RoMsamplecomp}.

With the reduced Robustness of Magic, we thus resolve the first of the two central obstacles in nonstabilizerness estimation: the requirement of tomographically complete data. The monotone $\mathrm{RoM}_\mathcal M$ can be evaluated using any experimentally accessible set of measurements, albeit with a strength that depends on the choice of $\mathcal M$. However, the second obstacle, the computational complexity, still remains at this point. Indeed, evaluating $\mathrm{RoM}_\mathcal M$ still requires access to a $V$-representation of the projected stabilizer polytope.

A straightforward upper bound on the size of the linear program in Eq.~\eqref{eq:reducedRoM} is obtained via a direct ``top--down'' construction,
\begin{equation}
    \mathrm{stab}(\mathcal M) \subseteq \Pi_\mathcal M(\mathrm{stab}_n)\;,
    \label{eq:topdownVrep}
\end{equation}
namely by projecting all $|\mathrm{stab}_n| = 2^{\Theta(n^2)}$ vertices of the full stabilizer polytope. This approach makes clear that, in the worst case, evaluating the reduced RoM is no easier than computing the original RoM: although the $\ell_1$-norm minimization can be solved in time polynomial in the number of constraints~\cite{vaidya1989speeding}, the number of vertices remains exponentially large. In the following, we show that by exploiting the algebraic structure of the measurement set $\mathcal M$, one can construct \emph{compressed} $V$-representations of the projected stabilizer polytope, thereby addressing the remaining computational bottleneck.

\subsection{Bottom-up reconstruction}

We will now show that there is a purely combinatorial construction of the reduced polytopes. First, define the frustration graph~\cite{gokhale2019minimizing,chapman2020characterization,mann2025graph,xu2025simultaneous} of $\mathcal M$ as the graph $G_\mathcal M$ with $V(G_\mathcal M) =\mathcal M$, and edges given by
\begin{equation}
    E(G_\mathcal M) \equiv \{\{P_i, P_j\} \in \mathcal M^{\times 2}| P_i P_j=-P_jP_i\}\;.
\end{equation}
Then, we can show:
\begin{thm}
    Let $\mathcal M$ be a set of $m$ $n$-qubit Pauli strings. Then, the reduced polytope has a $V$-representation given by
    \begin{equation}
        \mathrm{STAB}(\mathcal M) = \mathrm{conv}\{\vec v_{S,f}\}_{S \in I_\mathrm{max}(G_\mathcal M)\;,\;f \in \mathrm B_S}\;,
        \label{eq:reducedstabindependentsetss}
    \end{equation}
    where $I_\mathrm{max}(G_\mathcal M)$ is the set of (maximally) independent sets of $G_\mathcal M$, and
    \begin{equation}
        \mathrm B_S =\left\{\left .f \in \{\pm\}^S\right| \prod_{P \in S^\prime} {f_P} P\neq -1, \forall S^\prime \subseteq S\right\}\;,
        \label{admissiblef}
    \end{equation}
    denotes all bitstrings assignments over the independent set $S$ that do not have a negative Pauli product on the subsets, and the vector $\vec v_{S,f}$ is defined as
    \begin{equation}
    (\vec  v_{S,f})_P =  
    \begin{cases}
f_P \; \mathrm{if\;} f_PP \in S\;,\\
0\;\mathrm{otherwise.}
    \end{cases}
    \label{eq:consistentvector}
    \end{equation}
    Furthermore, the set $\mathrm{stab}(\mathcal M) =\{\vec{v}_{S,f}\}$ has size:
    \begin{equation}
        |\mathrm{stab}(\mathcal M)| \leq 2^{\min\{n,m\}+\min\{cm,(1/2+o(1))n^2\}}\;,
        \label{eq:upperboundreducedstab}
    \end{equation}
     where $c=\log 3/3$. Without further restrictions on $\mathcal{M}$, this bound is tight and reached by the set of marginals
     \begin{equation}
         \mathcal M= \bigcup_{i=1}^m \{X_i,Y_i,Z_i\}\;.
     \end{equation}
     \label{thm:reducedstabdecomp}
\end{thm}

We present the full proof in Appendix~\ref{proof:reducedstabdecomp}, and here we outline the main ideas underlying the construction. Starting from the original ``top--down'' $V$-representation of the reduced polytope in Eq.~\eqref{eq:topdownVrep}, one may view the extremal points of $\mathrm{STAB}(\mathcal M)$ as arising from the $\mathcal M$-correlations of pure stabilizer states. This perspective naturally leads to the question of which stabilizer structures are already encoded in the measurement set itself. Namely, which subsets of $\mathcal M$ can simultaneously fix expectation values $\langle P_i \rangle=\pm1$.

This intuition is formalized through the notion of \emph{closed commuting subsets} (CCSs)~\cite{sun2025stabilizerground}, defined as
\begin{equation}
    \mathrm C(\mathcal M)
    =
    \{Q \subseteq \pm \mathcal M \;|\;
    Q = \langle Q\rangle \cap \pm \mathcal M,\;
    -1 \notin \langle Q\rangle\}\;,
    \label{CCS}
\end{equation}
where $\pm \mathcal M = \mathcal M \sqcup (-\mathcal M)$ and $-\mathcal M=\{-P_1,-P_2,\cdots,-P_m\}$. The terminology reflects the two defining properties of any $Q \in \mathrm C(\mathcal M)$. First, the set is \emph{closed}, in the sense that the subgroup it generates, $\langle Q\rangle$, intersects $\pm \mathcal M$ only on the elements of $Q$ itself. Second, it is \emph{commuting}, since the condition $-1 \notin \langle Q\rangle$ implies that $\langle Q\rangle$ is an Abelian subgroup of the Pauli group and therefore defines a stabilizer group\footnote{To see why it also imposes commutativity, remember that a property of the Pauli group is that a Pauli pair either commutes or anti-commutes. But if there is an anticommuting pair $P, P^\prime \in Q$, note that $PQPQ=-1 \in \langle Q\rangle$. Hence, $-1 \notin \langle Q\rangle$ ensures that $\langle Q \rangle$ is abelian.}. Furthermore, we will consider \textit{maximal commuting subsets} (MCSs), where we impose the additional property that, given $Q\subseteq Q'$, then $Q=Q'$. We denote this set by $\mathrm C_{\mathrm{max}}(\mathcal M)$.

Operationally, the elements of $\mathrm C_{\mathrm{max}}(\mathcal M)$ may be interpreted as \emph{measurement contexts}~\cite{amaral2018graph}: maximal subsets of Pauli observables within $\mathcal M$ that can simultaneously attain definite values. The central and somewhat surprising result, proved in Appendix~\ref{proof:reducedstabdecomp}, is that these MCSs suffice to generate the reduced stabilizer polytope. More precisely, one can construct a $V$-representation such that
\begin{equation}
   \mathrm{stab}(\mathcal M) \cong \mathrm C_{\mathrm{max}}(\mathcal M)\;.
\end{equation}
This characterization makes explicit that the complexity of the reduced polytope is governed by the algebraic and combinatorial structure of the measurement set $\mathcal M$, rather than by the exponentially large set of pure stabilizer states appearing in the top-down construction.

The remainder of the proof establishes a connection between MCSs and a purely combinatorial problem: finding consistent independent sets in an associated frustration graph. This reformulation makes clear why we refer to the present approach as \emph{bottom--up}. Rather than projecting the full stabilizer polytope without regard to measurement structure, the stabilizer correlations accessible to a restricted measurement set are reconstructed directly from the relations among the measurements themselves. Moreover, the combinatorial expression of the convex hull in Eq.~\eqref{eq:reducedstabindependentsetss} naturally suggests an efficient algorithmic strategy: construct the frustration graph and enumerate its consistent independent sets.

We defer a detailed discussion of this algorithm to Sec.~\ref{sec:algoreducedstab}. Instead, in the next section, we illustrate the geometry of reduced stabilizer polytopes for small measurement sets, which provides valuable intuition for how their structure emerges from the underlying algebraic relations among the observables.

\subsection{On small polytopes \label{sec:smallreducedstab} and Pauli commutativity}

We now explore reduced polytopes for $m \leq 3$, since they are the case that can be visualized and hence easily understood. All the three cases are plotted in Fig.~\ref{fig:smallpolytopes}.

\begin{figure*}
    \centering
    \includegraphics{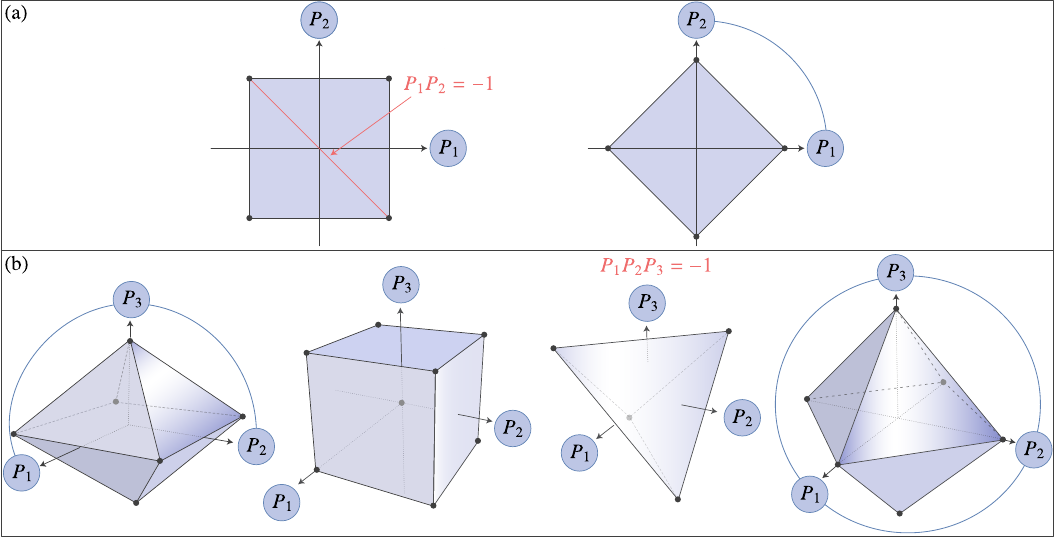}
    \caption{Illustration of some of the reduced polytopes for $m \in \{2,3\}$. We draw the polytopes in $[-1,1]^m$, with the axis labels as vertices of the frustration graph drawn in overlap. In (a), we show the two reduced polytopes coming from the two possible frustration graphs, where in red we show the corresponding reduced polytope if the extra condition $P_1P_2=-1$ is added in the commuting case. In (b), we highlight four possibilities: The fully commuting case, where the correlations span the full $[-1,1]^3$, the case where the commuting tuple also satisfies the constraint $P_1 P_2 P_3=-1$, the case of a graph of two edges, and the fully anti-commuting case, that reduces to the definition of $\mathrm{STAB}_1$.}
    \label{fig:smallpolytopes}
\end{figure*}

For $m=1$, we have $\mathcal M =\{P_1\}$. We have the frustration graph of one vertex, which also composes the single maximal independent subset. Since we can either assign $f_{P_1} =1$ or $f_{P_1}=-1$, we have:
\begin{equation}
    \mathrm{STAB}(\{P_1\}) = \mathrm{conv}(-1,1)=[-1,1]\;,
\end{equation}
which is indeed the full range of $\langle P_1\rangle$, since $\|P_1\|=1$. A single Pauli expectation value cannot witness nonstabilizerness.

The first non-trivial case is for $m=2$. If $[P_1,P_2] = 0$, the graph is fully disconnected and and there is a single maximally independent set, $I_\mathrm{max}(G_\mathcal M) = \mathcal M$. We have two further cases: If $P_1 \neq -P_2$, they are fully independent and:
\begin{equation}
    \mathrm{STAB}(\{P_1,P_2\}) =\mathrm{conv}(\{(\pm 1,\pm 1)\}) = [-1,1]^2\;.
\end{equation}
Otherwise, if $P_1=-P_2$, we obtain $\mathrm{conv}((1,-1),(-1,1))$, forming the red line in Fig.~\hyperref[fig:smallpolytopes]{\ref{fig:smallpolytopes}(a)}, since those two are the consistent sign assignments for the two Paulis. In all those cases, again, $\langle P_1\rangle$ and $\langle P_2\rangle$ fully obtain all the admissible expectation values, and no nonstabilizerness can be detected. The other case is when $\{P_1, P_2\} = 0$, and then the independent sets are formed by single Paulis. In this case, we have:
\begin{equation}
    \mathrm{STAB}(\{P_1,P_2\}) = \mathrm{conv}(\{(\pm 1, 0), (0,\pm 1 )\})\;,
\end{equation}
as illustrated as the second example in Fig.~\hyperref[fig:smallpolytopes]{\ref{fig:smallpolytopes}(a)}. Now, one can see that there is a range of correlations for which nonstabilizerness can be witnessed, since the corresponding polytope does not fully span the operator spectrum of both Paulis.

For $m=3$, different commutation structures can emerge. We might have $|E(G_\mathcal M)| \in \{0,1,2,3\}$, which labels the 4 possible graph permutation classes
based on the number of edges. The cases with $|E(G_\mathcal M)| \in \{0,2,3\}$ are illustrated in Fig.~\hyperref[fig:smallpolytopes]{\ref{fig:smallpolytopes}(b)}. Note that assuming negative Pauli product constraints, such as $P_1P_2=-1$, it would correspond also to a projection, since does correspond to imposing the extra constraint $\langle P_1 \rangle + \langle P_2\rangle = 0$. For $m \geq 3$ however, it starts to appear more interesting constraints, for example, $P_1P_2P_3=\pm 1$, since they cannot be written linearly in the expectation values. In the figure, the fully commuting case is shown with and without this extra constraint. Note that, for the fully anticommuting case, we have the octahedron that also defines $\mathrm{STAB}_1$, since indeed it is the reduced polytope for the single-qubit Pauli triple: $\mathrm{STAB}_1=\mathrm{STAB}(\{X,Y,Z\})$.

A very clear message that those examples show is that (independent) commuting sets of Paulis are trivial as witnesses. Indeed, we can show that this is a generic result:
\begin{prop}
    Let $\mathcal M =\{P_1, P_2, \cdots, P_m\}$ be a set of $n$-qubit Pauli strings that are indepenent, meaning that there is no $S \subseteq [m]$ such that $ \prod_{j \in S}P_j=\pm 1 $. If $\mathcal M$ is commuting,
    \begin{equation}
        \mathrm{STAB}(\mathcal M) =[-1,1]^{m}\;.
    \end{equation}
\end{prop}
\begin{proof}
    The proof follows from recognizing there is a single maximally independent subset, $I_\mathrm{max}(G_\mathcal M) = \{\mathcal M\}$. By independence, every $f \in \{\pm\}^{|\mathcal M|=m}$ is in $\mathrm B_\mathcal M$, since all of the sign assignments do satisfy the product condition. Hence, by Thm.~\ref{thm:reducedstabdecomp}, the reduced polytope is the convex hull of the Hamming cube:
    \begin{align}
    \mathrm{STAB}(\mathcal M) &=\mathrm{conv}\{\underbrace{(\pm 1, \pm 1, \cdots, \pm 1)}_{\mathrm{length\;}m}\}\;,  
    \end{align}
    which is exactly $[-1,1]^m$.
\end{proof}

\subsection{Explicit algorithm~\label{sec:algoreducedstab}}
Now that some intuition is gained on how the combinatorial problem of the frustration graph fixes the polytope, we present the general algorithm. The detailed analysis for the runtime is left for Appendix~\ref{proof:Vrepalgo}.
\begin{thm}
    Let $\mathcal M =\{P_1,P_2, \cdots, P_m\}$ be a set of $n$-qubit Pauli strings. There is a classical algorithm, sketched in Alg.~\ref{alg:Vrep}, that takes $\mathcal M$ as an input and outputs the $V$-representation of Thm.~\ref{thm:reducedstabdecomp} with runtime
    \begin{equation}
        T = p(n,m)2^{O(\min\{m,n\}+ \min(m,n^2))}\;,
    \end{equation}
    where $p(n,m)$ is bounded by some polynomial in $n$ and $m$.
    \label{thm:Vrepalgo}
\end{thm}
Note that the runtime of the algorithm to list all $|\mathrm{stab}(\mathcal M)|$ $m$-dimensional vectors has a similar upper bound as the length of the list itself. This is due the way Alg.~\ref{alg:Vrep} works: We list all the possible independent subsets and build the sign assignments, that give the exponential scaling. The polynomial overhead mostly comes from building the compatible sign assignments and the vectors $\vec v_{S,f}$.

\begin{algorithm}[H]
\caption{V-representation of STAB($\mathcal{M}$) \label{alg:Vrep}}
    \KwIn{$\mathcal{M} = \{P_1, P_2, \dots, P_m\}$ as the set of Pauli measurements;}
    \KwOut{$\mathrm{stab}(\mathcal M)=\{\vec{v}_1, \vec{v}_2, \dots, \vec{v}_{|\text{stab}(\mathcal{M})|}\}$ as defined on Thm.~\ref{thm:reducedstabdecomp}.}
    Initialize $V \gets \{\;\}$;\\
    Compute $E(G_\mathcal{M})$;\\
    Find $I_{max}(G_\mathcal{M})$;\\
    \For{$S \in I_\mathrm{max}(G_\mathcal M)$}{
        Compute all $f \in B_S$ and append $\vec{v}_{S,f}$ to $V$;
    }
    \Return{V}
\end{algorithm}
As a corollary, we get the complexity of estimating Eq.~\eqref{eq:reducedRoM}: It is known that from the theory of linear programs that finding a $\vec{x} \in \mathbb R^m$ that $\epsilon$-approximates a linear objective with $M$ constraints can be found in $O(m (m+ M)^{3/2}\log(1/\varepsilon))$ time~\cite{vaidya1989speeding}. Then, to estimate the reduced RoM, we first obtain the $V$-representation of the reduced polytope from the set of Pauli measurements via the algorithm of Thm.~\ref{thm:reducedstabdecomp} and then run the linear program with $M =O(2^{\min\{m,n\}+ \min\{cm,(1/2+o(1))n^2\}})$. Since the linear programming runtime dominates due to the extra power of $3/2$, we have:
\begin{cor}
    Given a set of expectation values of $n$-qubit Pauli strings $\mathcal M$ on a quantum state $\psi$, $(\langle P_1\rangle_\psi, \langle P_2\rangle_\psi, \cdots, \langle P_m \rangle_\psi)$, with $m \leq  \mathrm{poly}(n)$, there is an algorithm that $\epsilon$-approximates $\mathrm{RoM}_\mathcal M(\psi)$ in $O(2^{c^\prime \min\{m,n^2\}}\log(1/\epsilon))$ runtime, where $c^\prime=O(1)$. 
\end{cor}

Note that, up to a polynomial-time overhead, our algorithm outputs a compressed representation of the polytope in time proportional to its maximal possible size. Since the $2^{c m}$-length representation is known to be tight, this scaling indicates that the algorithm is near-optimal.

\section{stabilizer marginal problem and complexity lower bound \label{sec:stabmargprob}}
We will show in this section that under the assumption that $\mathrm P \neq \mathrm NP$, if there is an algorithm that decides if a set of Pauli correlations is inside the reduced polytope, it must run in
\begin{equation}
    T =\Omega(\mathrm{poly(n,m} ) )
    \label{eq:complexitylowerbound}
\end{equation}
time, meaning that there is no poly-time algorithm in the number of Pauli measurements to decide if the corresponding set of expectation values lies inside the reduced polytope or not. Doing so, we will make contact with a classical problem in quantum information theory, the quantum marginal problem.

First, we describe the task formulated as what is called a decision problem \footnote{In this section, since we are interested in the complexity-theoretic aspects, assume that all continuous-variable inputs are $\mathrm{poly}(n)$ sized. This includes the Pauli expectation values in Def. \ref{def:stabreducedprob} and the stabilizer probabilities in Def. \ref{def:stabmargprob}. This all can be done by assuming rational inputs.}:
\begin{defn}
    We define the reduced stabilizer membership problem as the decision problem (rSMP) defined by:
    \begin{itemize}
        \item \textbf{Inputs:} A set of $m  \leq \mathrm{poly}(n)$ $n$-qubit Pauli strings, $\mathcal M = \{P_1, P_2, \cdots, P_m\}$ and a vector $\vec{V} \in [-1,1]^m$;
        \item \textbf{Decision:} Output YES if $\vec V \in \mathrm{STAB}(\mathcal M)$ and NO otherwise.
    \end{itemize}
    \label{def:stabreducedprob}
\end{defn}
And, by Eq.~\eqref{eq:reducedstabwitness}, all the YES outputs to this problem match with all YES outputs with the stabilizer membership problem, which decides if a $n$-qubit quantum state $\psi$ is stabilizer or not. However, this variation only requires inputs of size $\mathrm{poly}(n)$.

Before proceeding, let us give a quick review what is called the quantum marginal problem (QMP). We assume to have $n$ qubits and subsets $\{R_1, R_2, \cdots, R_N\}\subseteq [n]$ of constant size, $|R_i|=k=O(1)$, and a list of density matrices $\{\rho_1, \rho_2, \cdots, \rho_N\}$ supported on those subsets. The problem is to find if there is a global density matrix $\rho$ such that it correctly marginalizes to the inputs within a $1/\mathrm{poly}(n)$ precision. It was shown that this problem is QMA-hard~\cite{liu20006consistency}, meaning that, even with quantum computers, it is expected that there is no efficient algorithm to solve it, under standard complexity-theoretic assumptions. 

The quantum marginal problem is the natural generalization of a classical problem in probability theory, referred to as the classical marginal problem (CMP). The difference is that the inputs are probability distributions $\{\vec p_1, \vec p_2, \cdots, \vec p_N\}$ each over a constant number of bits, and we ask if there is a global probability over $n$ bits that marginalizes to each of the inputs. It is known that this is NP-complete~\cite{pitowsky1989quantum, pitowsky1991correlation}.  

Note the close similarity with our task, namely, determining whether there exists a global stabilizer state that reproduces the statistics of a given set of Pauli observables $\mathcal{M}$. This motivates us to formulate a variant of the quantum marginal problem restricted to the stabilizer setting~\cite{hsieh2024resourcemarginal}:
\begin{defn}
    We define the stabilizer marginal problem (SMP) on $n$ qubits as the problem defined by:
    \begin{itemize}
        \item \textbf{Input: }A set of stabilizer density matrices $\{\phi_{1}, \phi_{2},\cdots, \phi_{N}\}$, with $N \leq \mathrm{poly}(n)$, $\phi_{i} \in \mathrm{STAB}_{|R_i|}$, specified by:
        \begin{equation}
            \phi_i =\sum_{\mathcal S \in \mathrm{stab}_{|R_i|}}p_\mathcal S\psi^\mathcal S \quad ;\quad \vec{p} \in \Delta_{\mathrm{stab}_{|R_i|}}\;,
        \end{equation}
        where $R_i \subseteq[n]\;,\; |R_i| \leq k=O(1)$ for all $i \in [N]$;
        \item \textbf{Decision: }Output YES if exists a $n$ qubit stabilizer state $\phi \in \mathrm{STAB}_n$ such that:
        \begin{equation}
            \tr_{\bar{R_i}}\phi =\phi_{i}\;,\;\forall i \in [N]\;,
        \end{equation}
        and NO otherwise.
    \end{itemize}
    \label{def:stabmargprob}
\end{defn}
The standard tool to relate problems of this kind is a polynomial-time reduction~\cite{arora2009computational}. Concretely, if a decision problem $A$ reduces to a problem $B$, then any algorithm that solves $B$ can be used to solve $A$ with only polynomial overhead. We establish a sequence of such reductions connecting the classical, stabilizer, and reduced membership problems:
\begin{thm}
    Denoting $\leq_\mathrm{poly}$ for the polynomial-time reduction, the following are true:
    \begin{enumerate}
        \item $\mathrm{CMP} \leq_\mathrm{poly}\mathrm{SMP}$;
        \item $\mathrm{SMP} \leq_\mathrm{poly} \mathrm{rSMP}$\;.
    \end{enumerate}
    Thus, $\mathrm{CMP} \leq_\mathrm{poly} \mathrm{rSMP}$, and therefore, both $\mathrm{SMP}$ and $\mathrm{rSMP}$ are NP-hard.
    \label{thm:CMPtorSMPreduction}
\end{thm}
The proof is described on Appendix~\ref{proof:CMPtorSMPreduction}. This shows that, by lower bounding the hardness of the SMP, one gets a hardness result for the reconstruction problem. This shows our claim in Eq.~\eqref{eq:complexitylowerbound}: If $\mathrm{P} \neq \mathrm{NP}$, there is no polynomial-time algorithm to solve an $\mathrm{NP}$-hard problem, and thus no way to solve the membership of the reduced stabilizer polytope in $\mathrm{poly}(n,m)$ time.

\section{Non-stabilizerness of Hamiltonian Ground states \label{sec:hamilton}}

In this section, we illustrate the power of our algorithm by numerically computing the reduced RoM of ground states of Hamiltonians, whose non-stabilizerness is recently discussed~\cite{oliviero2022magic, frau2024nonstabilizerness, timsina2025robustness}.

The general  approach will be as follows: A $n$ qubit Hamiltonian can be written as:
\begin{equation}
    H = \sum_{P \in \mathcal M^\prime}w_P P\;,
\end{equation}
where $\mathcal M^\prime$ is a set of Pauli operators and $\{w_P\}_{P \in \mathcal M^\prime}$ is a set of coefficients. Denote $\psi_\mathrm{gs}$ as a (or \emph{the}, in the unique case) corresponding ground state. In experiments and numerical simulations, one of the most natural observables is the ground state energy:
\begin{equation}
    E_0 = \langle H  \rangle_{\psi_\mathrm{gs}} = \sum_{P \in \mathcal M(H)} w_P \langle P \rangle_{\psi_\mathrm{gs}}\;.
\end{equation}
Indeed, we will assume that our measurement set is a subset $\mathcal M \subseteq \mathcal M(H)$, on the same footing as energy estimation, and is a weaker assumption when the inclusion is strict. For concreteness, we will focus with translationally-invariant, one-dimensional $k \leq 2$-local Hamiltonians~\cite{zeng2019quantum}, that have the form:
\begin{equation}
    H =\sum_{\alpha,\beta \in \{0,1,2,3\}}\sum_i J^{\alpha \beta} X^\alpha_i X^\beta_{i+1}\;,
\end{equation}
where we referred to the components of $\mathbf (X^0_i, X^1_i, X^2_i, X^3_i) =(1, X_i, Y_i, Z_i)$ in the sum. In particular, we will focus in the following family of examples:
\begin{align}
    H_\mathrm{ANNNI} &= -\sum_{i}(  Z_i Z_{i+1}- k Z_{i}Z_{i+2} + g X_i)\;,\label{eq:ANNNI}\\
    H_\mathrm{XXZ} &= \frac{1}{4} \sum_{n} \left( X_n X_{n+1} + Y_n Y_{n+1} + \Delta Z_n Z_{n+1} \right) - \frac{h}{2} \sum_{n} X_n\;,\label{eq:XXZ}
\end{align}
where all the couplings are assumed to be positive. The first Hamiltonian is referred to as the axial next-nearest neighbor Ising (ANNNI) model~\cite{selke1988annni} and the second is referred as XXZ chain (with a transverse field) ~\cite{dmitriev2002gap}. As discussed in the references above, both Hamiltonian families exhibit quantum phase transitions~\cite{sachdev2011quantum}: The ground state exhibits abrupt changes as the parameters are varied. A natural signature is to look at the energy gap:
\begin{equation}
    \Delta(H)  =\lim_{n\to \infty}(E_1-E_0)\;,
\end{equation}
where $E_1$ is the first excited state. A ground state is said to be \emph{gapped} if $\Delta(H) >0$, and \emph{gapless} if $\Delta(H) = 0$. A \emph{phase diagram} is constructed by identifying states that have their states smoothly varying with the parameters that define the Hamiltonian, although this definition can only be made mathematically sharp in the gapped case~\cite{zeng2019quantum}. Numerically, this can be probed by looking at the behaviour of observables as one takes the system size to be large.
\begin{figure}
    \centering
    \includegraphics{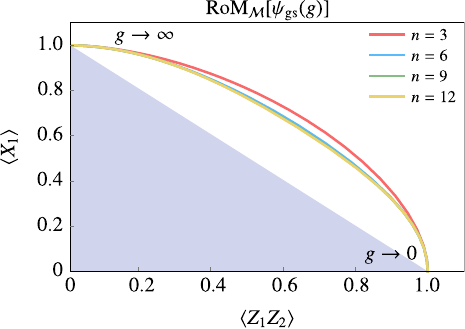}
    \caption{Parametric evolution of the ground-state expectation values and the reduced Robustness of Magic ($\text{RoM}_{\mathcal{M}}$) for the transverse-field Ising model (TFIM). The plot explicitly illustrates only the first quadrant of the observable plane. The shaded region delineates the boundary of the reduced stabilizer polytope associated with the local measurement set $\mathcal{M}=\{Z_1Z_2, X_1\}$. The solid curves denote the trajectory of the ground state, $\psi_{\text{gs}}(g)$, within the observable space $(\langle Z_1Z_2 \rangle, \langle X_1 \rangle)$ as the transverse field $g$ varies from $0$ to $\infty$ across varying system sizes ($n=3, 6, 9, 12$). As the system size $n$ increases, the maximum value of the reduced robustness concentrates around the critical point $g=1$.}
    \label{fig:Ising}
\end{figure}
Let us consider the case of the  ANNNI. For $k=0$, the corresponding Hamiltonian is known as the transverse-field Ising model (TFIM). Its nonstabilizerness was analyzed in~\cite{oliviero2022magic} by computing the stabilizer Rényi entropy (SRE)~\cite{leone2022stabilizer}. In general, it is necessary to have full tomographic information compute the SRE. It turns out that this model hosts a phase transition at $g=1$ between two gapped phases, that have stabilizer points: Notice that $\psi_\mathrm{gs}(g \to \infty) = (\ket{+}\bra{+})^{\otimes n}$ is the unique ground state at large fields, stabilized by the group $\langle X_1, X_2, \cdots, X_n\rangle$ and $\psi_\mathrm{gs}(g=0)$ is any state formed through the coherent superposition of $\ket{0}^{\otimes n}$ and $\ket{1}^{\otimes n}$, with stabilizer group $\langle Z_1 Z_2, Z_2Z_3, \cdots\rangle$. The transition is a gapless point that tunes into the phases that are adiabatically connected between the two points. The SRE follows this behaviour: In the stabilizer points, it vanishes, and acquires a maximum value in the transition between the phases that are connected with the stabilizer points.

Using our approach, we observe that this behavior can be found already with exponentially fewer measurements. Consider:
\begin{equation}
\mathcal  M= \{Z_1Z_2, X_1\}\;,
\end{equation}
which corresponds to energy measurements on the first pair of qubits. The corresponding reduced polytope is the one corresponding to the second graph in Fig. \ref{fig:smallpolytopes} (a). Solving the linear program of Eq.~\eqref{eq:reducedRoM}, we can compute $\mathrm{RoM}_\mathcal M[\psi_\mathrm{gs}(g)]$, plotted in Fig.~\ref{fig:Ising}. As $n$ increases, we can see that the maximum robustness remains around the critical value $g=1$, indicating the phase transition.

We can extend this example to the full phase diagram of the models in Eqs.~\eqref{eq:ANNNI}~and~\eqref{eq:XXZ}. In the optimal case where we measure all the $\Theta(n)$ Paulis in the Hamiltonian, meaning that we compute:
\begin{align}
    \mathcal M(H_\mathrm{ANNNI}) &\mapsto \mathrm{RoM}_{\mathcal M(H_\mathrm{ANNNI})}[\psi_\mathrm{gs}(k,g)]\;, \\
    \mathcal M(H_\mathrm{XXZ}) &\mapsto \mathrm{RoM}_{\mathcal M(H_\mathrm{XXZ})}[\psi_\mathrm{gs}(h, \Delta)]\;.
\end{align}
The results are shown in Fig.~\ref{fig:ANNNi-and-XXZ}. As one can see, the behavior of the energy gap has a similar behaviour as the nonstabilizerness as captured by the reduced RoM, where gapless regions are expected to have high nonstabilizerness and gapped phases that hosts stabilizer points, as is the case for the TFIM, has low nonstabilizerness near stabilizer regions. Although most of the phase diagrams shows nonstabilizerness, there are particular regions that it vanishes. For the ANNNI, note that the gapped region at $g=0$ has vanishing RoM. This is due to the fact that the ground state is in the code stabilized by the group generated by the $Z \otimes Z$ Pauli pairs on first and second neighbors. Similarly, for the XXZ, in the region where $h=0$ and we have the ferromagnetic coupling $\Delta \in [-2,-1]$, has stabilizer ground states. 

We stress that although the results were derived by exact diagonalization of the hamiltonians for few qubits, we could use data coming from quantum simulators/computers and more efficient classical simulation algorithms to run the same algorithm for a larger number of qubits, which is a bottleneck for other measures, as previously mentioned. For example, in \cite{timsina2025robustness}, the RoM was studied for $\leq 8$-size subsystems of quantum Gibbs states, and altough our reduced form does not correspond to the full information of the RoM, the numerical experiements above indicate that they capture much of the same information with lower computational algorithmic cost.

\begin{figure}
    \centering
    \includegraphics{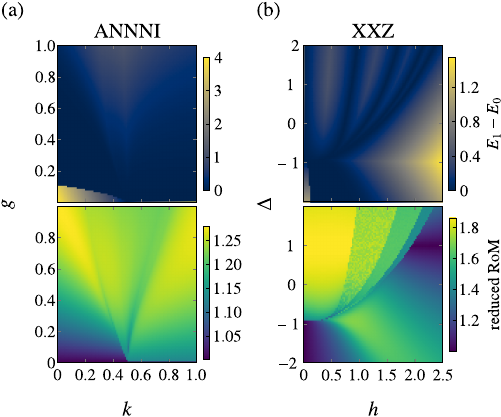}
    \caption{Comparison between the energy gap and the Reduced RoM measure for the ANNNI and XXZ models. For the ANNNI model, 10 qubits were used, while 9 qubits were used for XXZ model. Column (a) presents the results for the ANNNI model as function of the frustration $k$ and the transverse field $g$. Column (b) details the XXZ model in terms of the magnetic field $h$ and the anisotropy $\Delta$.}
    \label{fig:ANNNi-and-XXZ}
\end{figure}

\section{Conclusion and final remarks}
\label{sec:discuss}

In this work, we have shown how nonstabilizerness can be witnessed and quantified even when the number of accessible measurements is limited. By restricting stabilizer resource theory to a subset of observables, we constructed an algorithm that takes as input the expectation values of Pauli operators and decides whether they are compatible with stabilizer statistics. We demonstrated the practical performance of this approach through explicit many-body examples, and established its near-optimality from a complexity-theoretic perspective: the superpolynomial dependence on the number of measurements is not an artifact of the method, but is in fact unavoidable unless widely believed complexity-theoretic assumptions fail.

Importantly, the reduced polytope naturally defines a nonstabilizerness monotone, the reduced robustness of magic $\mathrm{RoM}_{\mathcal M}$, which lower bounds the magic content of the full state. Operationally, this immediately translates into a lower bound on the classical simulation cost: measuring only $m \leq \mathrm{poly}(n)$ Pauli observables already suffices to infer a quasiprobability sample-complexity bound. Hence, even limited measurement data can certify nontrivial classical hardness.

Several directions for future research naturally emerge. One concerns connections to quantum learning theory, in particular quantum PAC learning, where the goal is to predict expectation values drawn from a distribution using only a linear number of measurements~\cite{aaronson2007learnability}. It was shown in~\cite{rocchetto2017stabiliser} that stabilizer states are efficiently PAC-learnable, in contrast to generic quantum states. An interesting open question is how this learnability behavior extends to families of states with bounded nonstabilizerness, and whether the tools developed here can help characterize the sample complexity as a function of the magic content. Another promising direction is to relate our framework to recent advances on quantum correlation sets of Pauli observables~\cite{xu2024bounding,xu2025simultaneous}, which are likewise governed by algebraic properties of the underlying frustration graph. Understanding these connections may yield a more unified geometric and combinatorial picture of restricted quantum correlations.

Finally, it is natural to explore potential implications for fault-tolerant quantum computation, particularly in the context of magic-state injection~\cite{bravyi2005universal}. Since the robustness of magic (RoM) provides quantitative bounds on resources in unitary synthesis~\cite{howard2017application}, it would be interesting to investigate whether our reduced-polytope approach can lead to new resource-estimation or certification protocols tailored to experimentally accessible observables.

Overall, our results offer a new perspective on the geometry of nonstabilizerness and provide a practical route toward its characterization in large-scale quantum systems. More broadly, they suggest that the complexity inherent in quantum resource theories need not be confronted in its full exponential form, but can instead be meaningfully controlled by projecting the problem onto experimentally and computationally accessible subspaces.

\begin{acknowledgments}
We acknowledge funding from Coordenação de Aperfeiçoamento de Pessoal de Nível Superior – Brasil (CAPES) – Finance Code 001, the Simons Foundation (Grant No. 1023171, R.C.), the Brazilian National Council for Scientific and Technological Development (CNPq, Grants No.403181/2024-0, and 301687/2025-0), the Financiadora de Estudos e Projetos (Grant No. 1699/24 IIF-FINEP),  a guest professorship from the Otto M\o nsted Foundation,  the Danish National Research Foundation grant bigQ (DNRF 142) and VILLUM FONDEN through a research grant (40864). We thank the High-Performance Computing Center (NPAD) at UFRN for providing computational resources. We also thank the Qiskit team for their SDK, which was utilized in our simulations.
\end{acknowledgments}

\bibliography{refs}

\appendix
\onecolumngrid

\section{Short proofs and aiding Propositions}

\noindent 
A result that we use a lot is the following proposition:

\begin{prop}
    Let $\mathcal S \subseteq \mathcal P_n$ a stabilizer group of generic rank, and let $\psi^\mathcal S = 2^{-n}\sum_{s \in \mathcal S}s$ be the stabilizer projector. Then, for any $P \in \mathcal P_n$:
    \begin{equation}
       \langle P \rangle_{\psi^\mathcal S} = \tr(P \psi^\mathcal S) = (P, \mathcal S)\equiv  1[P \in \mathcal S]- 1[-P \in \mathcal S]\;, 
    \end{equation}
    where $1[\cdot]$ is the indicator function. Note that $(P, \mathcal S) \in \{-1,0,1\}$.
    \label{prop:paulistab}
\end{prop}
\begin{proof}
    It follows by the definition of $\psi^\mathcal S$ and the Hilbert-Schmidt orthogonality of Pauli matrices, which imposes $\tr(PQ) =\delta_{PQ}2^n$ for all $P, Q \in \mathcal P_n$:
    \begin{equation}
        \tr(P \psi^\mathcal S) = \frac{1}{2^{n}}\sum_{s \in \mathcal S} \tr(Ps) = \sum_{s \in \mathcal S}(\delta_{P,s}-\delta_{-P,s})\;,
    \end{equation}
from which the results follows by recognizing $1[\pm P \in \mathcal S] = \sum_{s \in \mathcal S}(\delta_{P,s}-\delta_{-P,s})$.
\end{proof}

This result implies the concentration of Pauli strings of stabilizer pseudomixtures/mixtures: Let $\psi =\sum_\mathcal S x_\mathcal S \psi^\mathcal S$, and let $P$ be a Pauli string. Then, by the Proposition above:
\begin{equation}
    \langle P\rangle_\psi = \sum_{\mathcal S \in \mathrm{stab}_n}x_\mathcal S (P, \mathcal S) = \sum_{\substack{\mathcal S \in \mathrm{stab}_n\\ P \in \pm \mathcal S}}x_\mathcal S (P, \mathcal S)\;,
    \label{eq:pauliexpecpseudo}
\end{equation}
where now the sum is restricted to stabilizer groups that include either $P$ or $-P$.

Yet another result is the 1-norm bound on changing $V$-representations of a polytope:

\begin{prop}
    Let $P \subseteq [-1,1]^N$ be a polytope that can be written as the convex hull of two sets:
    \begin{equation}
        P = \mathrm{conv}(\mathbf {u}_1, \mathbf {u}_2, \cdots, \mathbf {u}_{k_1}) = \mathrm{conv}(\mathbf {v}_1, \mathbf {v}_2, \cdots, \mathbf {v}_{k_2})\;,
    \end{equation}
    furthermore, assume that $\{\mathbf {u}_1, \mathbf {u}_2, \cdots, \mathbf {u}_{k_1}\}\subset \mathrm{conv}(\mathbf {v}_1, \mathbf {v}_2, \cdots, \mathbf{v}_{k_2})$. Then, given $\mathbf v = \sum_{i=1}^{k_1}\alpha_i \mathbf u_i = \sum_{j=1}^{k_2}\beta_j \mathbf v_j \in P$, we have:
    \begin{equation}
        \|\boldsymbol{\beta}\|_1=\sum_{j=1}^{k_2}|\beta_j|\; \leq \|\boldsymbol{\alpha}\|_1=\sum_{i=1}^{k_1}|\alpha_i|.
    \end{equation}
    \label{prop:1norminequality}
\end{prop}
\begin{proof}
Given that the $\mathbf u$'s are in the convex hull of the $\mathbf v$'s:
\begin{equation}
    \mathbf {u}_i = \sum_{j=1}^{k_2}q^i_j\mathbf{v}_j \quad ;\quad \mathbf q^i \in \Delta_{k_2} \;,\;\forall i \in [k_1]\;.
\end{equation}
It follows that $\beta_j = \sum_{i=1}^{k_1} q^i_j \alpha_i$. A straightforward application of the triangle inequality finalizes the proof:
\begin{equation}
    \|\boldsymbol{\beta}\|_1 = \sum_{j=1}^{k_2}|\beta_j| \leq \sum_{i=1}^{k_1}\sum_{j=1}^{k_2} q^i_j |\alpha_i| =\|\boldsymbol{\alpha}\|_1\;,
\end{equation}
where on the last line we used that $\sum_{j=1}^{k_2}q^i_j=1$ for all $i \in [k_1]$ and the definition of the 1-norm of $\boldsymbol{\alpha}.$
\end{proof}

\subsection{Proof of Lemma \ref{lem:weakdepstab}\label{proof:weakdepstab}}
\noindent
Let us denote $\mathcal M =\{P_i\}_{i=1}^m$ and $\mathcal M \otimes \mathbb I =\{P_i \otimes \mathbb I\}_{i=1}^m$. Let $\Pi_\mathcal M\vec{\psi}\in \mathrm{STAB}_n(\mathcal M)$. Then:
\begin{equation}
    \Pi_\mathcal M\vec{\psi}= \sum_{\mathcal S \in \mathrm{stab}_n} p_\mathcal S \Pi_\mathcal M\vec{\psi^\mathcal S}\;,
\end{equation}
where $\Pi_\mathcal M\vec{\psi^\mathcal S} = (\langle P_1 \rangle_{\psi^\mathcal S}, \langle P_2 \rangle_{\psi^\mathcal S}, \cdots, \langle P_m \rangle_{\psi^\mathcal S} )$. By tensoring with the identity, we construct the stabilizer group $\mathcal S \otimes \mathbb I$ by tensoring all Paulis in $\mathcal S$ with $\mathbb I_{n^\prime-n}$, which is not full rank. Two facts hold:
\begin{align}
    \psi^{\mathcal S \otimes \mathbb I} &= \frac{1}{2^n}\sum_{s\otimes \mathbb I_{n^\prime-n} \in \mathcal S\otimes \mathbb I}s\otimes \mathbb I_{n^\prime-n} \in \mathrm{STAB}_{n^\prime}\;, \label{eq:notfullrankstab}\\
    \langle P_i\rangle_{\psi^\mathcal S} &= \langle P_i\otimes \mathbb I_{n^\prime-n}\rangle_{\psi^{\mathcal S \otimes \mathbb I}} \;,\; \forall i \in \{1,2,\cdots, m\}\;. \label{eq:paulipromote}
\end{align}
Eq.~\eqref{eq:paulipromote} follows from Eq.~\eqref{eq:notfullrankstab}, and the first equation follows from the fact that $\psi^{\mathcal S \otimes \mathbb I}$ is the (proportionally normalized) projector into the subspace of states in $\mathbb C^{2\otimes n^\prime}$, which can be written as uniform distribution over its pure stabilizer basis~\cite{gottesman2024surviving}. By Eq.~\eqref{eq:paulipromote}, it follows that $\Pi_\mathcal M \vec{\psi^\mathcal S} =  \Pi_\mathcal M \vec{\psi}^{\mathcal S\otimes \mathbb I}$, thus showing that $\Pi_\mathcal M \vec{\psi} \in \mathrm{STAB}_{n^\prime}(\mathcal M\otimes \mathbb I)$, and therefore, $\mathrm{STAB}_n(\mathcal M) \subseteq \mathrm{STAB}_{n^\prime}(\mathcal M \otimes \mathbb I)$.

Now, consider the other way around. Let $\Pi_\mathcal M \vec{\psi^\prime} \in \mathrm{STAB}_n(\mathcal M \otimes 
 \mathbb I)$. Then, it admits a decomposition:
 \begin{equation}
     \Pi_\mathcal M \vec{\psi^\prime} =\sum_{\mathcal S^\prime \in \mathrm{stab}_{n^\prime}} p_{\mathcal S^\prime} \Pi_\mathcal M\vec{\psi^{\mathcal S^\prime}}\;.
 \end{equation}
Now, consider a tensor splitting $\mathbb C^{2\otimes n^\prime} \cong \mathbb C^{2\otimes n}\otimes \mathbb C^{2\otimes (n^\prime-n)}$. For any $\mathcal S^\prime \in \mathrm{stab}_{n^\prime}$, we can construct the following stabilizer group:
\begin{equation}
    \mathcal S^\prime|_{n} = \{s^\prime \in \mathcal S^\prime |s^\prime =s \otimes \mathbb I_{n^\prime-n}\} \in \mathrm{stab}_n\;,
\end{equation}
meaning that it acts non-trivially just on the first tensor factor. Note that imposing this support condition on $s^\prime \in \mathcal S^\prime$ is the same as imposing $\tr_{n^\prime-n}(s) \neq 0$, due to Paulis being traceless. Finally, we note that:
\begin{equation}
    \langle P_i \otimes \mathbb I_{n^\prime-n}\rangle_{\psi^{\mathcal S^\prime}} = \frac{1}{2^n}\sum_{s^\prime \in \mathcal S^\prime} \tr(P_i \otimes \mathbb I_{n^\prime-n}s^\prime) = \frac{1}{2^n}\sum_{s \in \mathcal S^\prime|_n} \tr(P_i s) =\langle P_i \rangle_{\psi^{\mathcal S^\prime|_n}}\;,
\end{equation}
since, for stabilizers in $\mathcal S^\prime \backslash \mathcal S^\prime|_{n}$ will have vanishing contribution. Therefore, this shows that $\Pi_{\mathcal M} \vec{\psi}^{\mathcal S^\prime} = \Pi_{\mathcal M} \vec{\psi}^{\mathcal S^\prime|_n}$, showing the inclusion $\Pi_\mathcal M \psi^\prime \in \mathrm{STAB}_n(\mathcal M)$, and therefore, $\mathrm{STAB}_{n^\prime}(\mathcal M \otimes \mathbb I)\subseteq \mathrm{STAB}_n(\mathcal M)$. $\square$

\subsection{Proof of Lemma \ref{lem:reducedROMproperties}\label{proof:reducedromproperties}}
\noindent 
Let us go one-by-one:
\begin{enumerate}
    \item The inequality follows as, given any $\mathbf x \in \mathbb R^{|\mathrm{stab}(\mathcal M)|}$,
    \begin{equation}
        \|\mathbf x\|_1 = \sum_{j=1}^{|\mathrm{stab}(\mathcal M)|}|x_j| \geq \left |\sum_{j=1}^{|\mathrm{stab}(\mathcal M)|} x_j\right|=1\;,
    \end{equation}
    by the constraint. If there is a stabilizer state $\phi \in \mathrm{STAB}_n$ agreeing with $\psi$ on all Paulis of $\mathcal M$, it follows that $\mathrm{RoM}_\mathcal M(\psi) =\mathrm{RoM}_\mathcal M(\phi)=\|\mathbf x\|_1=\sum_{j} x_j=1$, since $\phi$ is a stabilizer mixture. On the other hand, note that the reverse triangle inequality above is tight iff all $x_j\geq 0$. However, if that is the case, note that we have:
    \begin{equation}
        \Pi_\mathcal M \vec{\psi} =\Pi _\mathcal M\left(\sum_{j=1}^{|\mathrm{stab}(\mathcal M)|} \sum_{\mathcal S \in \mathrm{stab}_n}x_j p^j_\mathcal S \vec{\psi}^\mathcal S\right) \equiv \Pi_\mathcal M \vec{\phi}\;,
    \end{equation}
    where we defined $\phi =\sum_{j, \mathcal S}x_j p^j_\mathcal S \psi^\mathcal S \in \mathrm{STAB}_n$, since $\mathbf x \in \Delta_{|\mathrm{stab}(\mathcal M)|}$ and $\mathbf p^j \in \Delta_{|\mathrm{stab}_n|}$ for all $j \in [|\mathrm{stab}(\mathcal M)|]$.
    \item If $\mathcal E$ is a TP stabilizer operation:
    \begin{equation}
        \mathcal E(\psi^\mathcal S) =\sum_{\mathcal S^\prime \in \mathrm{stab}_n}p^\mathcal S_{\mathcal S^\prime} \psi^{\mathcal S^\prime} \quad, \quad \mathbf p^\mathcal S \in \Delta_{|\mathrm{stab}_n|}\;.
    \end{equation}
    Hence, given a stabilizer pseudomixture of $\psi = \sum_{\mathcal S \in \mathrm{stab}_n} x_\mathcal S \psi^\mathcal S$, we have that $\Pi_\mathcal M\vec{\mathcal E(\psi)}=\sum_{\mathcal S, \mathcal S^\prime} p^\mathcal S_{\mathcal S^\prime}x_\mathcal S \Pi_\mathcal M \vec{\psi}^{\mathcal S^\prime}$\;, thus, using that the $\mathcal M$-correlations of pure stabilizer states are in $ \mathrm{conv}(\mathrm{stab}(\mathcal M))$, we have, by the bound of Proposition \ref{prop:1norminequality}:
    \begin{equation}
        \mathrm{RoM}_\mathcal M(\mathcal E(\psi)) \leq \sum_{\mathcal S, \mathcal S^\prime \in \mathrm{stab}_n} |x_\mathcal S| p^\mathcal S_{\mathcal S^\prime} =\|\mathbf x\|_1\;,
    \end{equation}
    which, upon minimizing over $\mathbf x$, arrive at $\mathrm{RoM}_\mathcal M (\mathcal E(\psi)) \leq \mathrm{RoM}_\mathcal M(\psi)$.
    \item Submultiplicative follows if one writes the stabilizer pseudomixture of the parts:
    \begin{equation}
        \psi  =\sum_{\mathcal S \in \mathrm{stab}_n}x_\mathcal S \psi^\mathcal S \quad ; \quad \rho = \sum_{\mathcal S \in \mathrm{stab}_n}y_\mathcal S \psi^\mathcal S\;,
     \end{equation}
     and noting that $\psi \otimes \rho = \sum_{\mathcal S, \mathcal S^\prime }x_\mathcal S y_{\mathcal S^\prime} \psi^\mathcal S \otimes \psi^{\mathcal S^\prime}$ defines itself a pseudomixture of the product. Then, again applying Proposition \ref{prop:1norminequality} to $\Pi_\mathcal M\vec{\psi\otimes \rho} = \sum_{\mathcal S, \mathcal S^\prime} x_\mathcal S y_{\mathcal S^\prime}\Pi_\mathcal M\vec{\psi^\mathcal S \otimes \psi^{\mathcal S^\prime}}$, we arrive at:
     \begin{equation}
         \mathrm{RoM}_\mathcal M(\psi \otimes \rho) \leq \sum_{\mathcal S, \mathcal S^\prime}|x_\mathcal S||y_{\mathcal S^\prime}| =\|\mathbf x\|_1 \|\mathbf y\|_1\;,
     \end{equation}
     which follows when optimizing over $\mathbf x$ and $\mathbf y$ in the rhs. The proof for convexity follows by the same strategy: consider $\Pi_\mathcal M \vec{\psi}_j=\sum_\mathcal S x_{\mathcal S, j} \Pi_\mathcal M \vec{\psi}^\mathcal S$, which will give an upper bound:
     \begin{equation}
         \mathrm{RoM}_\mathcal M\left( \sum_j p_j \psi_j\right) \leq \left\|\sum_j p_j \mathbf x_j\right\|_1 \leq \sum_j |p_j | \|\mathbf x_j\|_1\;,
     \end{equation}
     yielding $\mathrm{RoM}_\mathcal M(\sum_j p_j \psi_j) \leq \sum_j |p_j| \mathrm{RoM}_\mathcal M(\psi_j)$ upon optimization.
     \item The bound follows from the fact that we have two mixture decompositions:
     \begin{equation}
         \Pi_\mathcal M \vec{\psi}  = \sum_{j=1}^{|\mathrm{stab}(\mathcal M)|} x_j \mathbf v_j =\sum_{\mathcal S \in \mathrm{stab}_n } x_\mathcal S \Pi_\mathcal M \vec{\psi}^\mathcal S\;,
     \end{equation}
     with $\mathbf v_j \in \mathrm{conv}\{\Pi_\mathcal M \vec{\psi}^\mathcal S\}_{\mathcal S \in \mathrm{stab}_n}$, for all $j$. Therefore, by Proposition \ref{prop:1norminequality}, all such decompositions satisfy:
     \begin{equation}
        \sum_{j=1}^{|\mathrm{stab}(\mathcal M)|}|x_j| \leq \sum_{\mathcal S \in \mathrm{stab}_n}|x_\mathcal S| \quad \Rightarrow \quad \mathrm{RoM}_\mathcal M(\psi) \leq \mathrm{RoM}(\psi)\;.
     \end{equation}
\end{enumerate}
$\square$

\section{Proof of Theorem \ref{thm:reducedstabdecomp} \label{proof:reducedstabdecomp}}
First, we will show the relation to the maximally commuting subsets defined on Eq.\eqref{CCS}, and then do the graph-theoretic interpretation. By Eq.~\eqref{eq:pauliexpecpseudo}, we know that for any given stabilizer state $\phi = \sum_{\mathcal S \in \mathrm{stab}_n}p_\mathcal S \psi^\mathcal S \in \mathrm{STAB}_n$, its projection has the form:
\begin{equation}
    \Pi_\mathcal M \vec{\phi} = \sum_{\substack{\mathcal S \in \mathrm{stab}_n\\\mathcal S \cap \pm \mathcal M \neq \varnothing}}p_\mathcal S \Pi_\mathcal M \vec{\psi}^\mathcal S\;,
    \label{eq:measurementstabmixture}
\end{equation}
since $\Pi_\mathcal M \vec{\psi}^{\mathcal S: \mathcal S\cap \pm \mathcal M =\varnothing} = \vec{0}$. Now, let us define $\mathcal S_1\sim\mathcal S_2$ in $\mathrm{stab}_n$ iff $\mathcal S_1 \cap \pm \mathcal M = \mathcal S_2 \cap \pm \mathcal M$; it is easy to see that this is an equivalence relation. Then, due to Proposition \ref{prop:paulistab}, it follows that:
\begin{equation}
    \langle P_i \rangle_{\psi^{\mathcal S_1}} = \langle P_i \rangle_{\psi^{\mathcal S_2}} \;,\; \forall i \in [m] \quad \Rightarrow \quad \Pi_\mathcal M \vec{\psi}^{\mathcal S_1} = \Pi_\mathcal M \vec{\psi}^{\mathcal S_2}\;.
\end{equation}
Thus, the distinct elements on the convex decomposition of Eq.~\eqref{eq:measurementstabmixture} is for stabilizer groups that have distinct intersections with the measurement set. We can rewrite Eq.~\eqref{eq:measurementstabmixture} as:
\begin{equation}
    \Pi_\mathcal M \vec{\phi} = \sum_{[\mathcal{S}] \in \mathrm{stab}_n|_\mathcal M }p_{[\mathcal{S}]} \vec{v}_{[\mathcal{S}]}\;,\; \vec{p}\in \Delta_{|\mathrm{stab}_n|_\mathcal M|}\;,
    \label{eq:convstabM}
\end{equation}
where $\mathrm{stab}_n|_\mathcal M\equiv \left(\mathrm{stab}_n/\sim \right) \setminus[\varnothing]$ is the set of all equivalence classes except the stabilizers with vanishing intersection with $\mathcal{M}$, $p_{[\mathcal{S}]}=\sum_{\mathcal{S}^\prime\in [\mathcal{S}]}p_{\mathcal{S}^\prime}$  and the corresponding vector is:
    \begin{equation}
    (\vec  v_{[\mathcal{S}]})_i =  
    \begin{cases}
1 \; \mathrm{if\;} P_i \in \mathcal{S}\;,\\
-1 \; \mathrm{if\;}-P_i \in \mathcal{S}\;,\\
0\;\mathrm{otherwise.}
    \end{cases}
    \end{equation}

\textbf{Claim:} The correspondence $[\mathcal{S}]\in\mathrm{stab}_n|_\mathcal  M\mapsto \mathcal{S}\cap\pm\mathcal{M}\in\mathrm{C}_\mathrm{max}(\mathcal M)$ is a bijection. 

First, by consequence of the equivalence relation, it is well-defined; since $\mathcal{S}$ is a stabilizer group, we have $-\mathds{1}\notin\mathcal{S}\cap\pm\mathcal{M}$ and since $\mathcal{S}$ is full rank, $\mathcal{S}\cap\pm\mathcal{M}$ is maximal, i.e., $\mathcal{S}\cap\pm\mathcal{M}\in\mathrm{C}_\mathrm{max}(\mathcal M)$.

Given any $Q\in\mathrm{C}_\mathrm{max}(\mathcal M)$, since it is a commuting set of Paulis, we can always find a maximal rank stabilizer group $\mathcal{S}\supseteq Q$. For any choice of $\mathcal{S}$, we have $\mathcal{S}\cap\pm\mathcal{M}=Q$; if this were not the case, say, $\mathcal{S}\cap\pm\mathcal{M}\supsetneq Q$, we could take some $P\in \mathcal{S}\cap\pm\mathcal{M}\setminus Q$ and add it to $Q$, rendering $Q$ not maximal. Thus, the inverse map is well-defined.

By abuse of notation, we can identify $\vec{v}_Q=\vec{v}_\mathcal{S}$ via the aforementioned bijection. Therefore by Eq.~\eqref{eq:convstabM}:
\begin{equation}
    \mathrm{STAB}(\mathcal M) =\mathrm{conv}\{\vec{v}_Q\}_{Q \in \mathrm C_{\mathrm{max}}(\mathcal M)}\;.
    \label{eq:reducedstabcommuting}
\end{equation}

We can further strengthen this claim: Each $\vec{v}_Q$ is actually an extremal point of $\mathrm{STAB}(\mathcal M)$, and thus $\mathrm{stab}_n|_\mathcal M$ also correspond to extremal vertices.

Suppose that $\vec{v}_Q=\sum_{Q^\prime}p_{Q^\prime} \vec{v}_{Q^\prime}$ with $p_{Q^\prime}>0$. We claim that $Q\subseteq Q^\prime$ for all $Q^\prime$ in the sum. If there was some $Q\supsetneq Q^\prime$, we could choose some $P_i\in Q^\prime\setminus Q$, implying $(\vec{v}_Q)_i=0< (\vec{v}_{Q^\prime})_i$, which is absurd.

Then $Q\subseteq Q^\prime$ for all $Q^\prime$, and by maximality, all $Q^\prime =Q$ and the convex combination is trivial.

We now construct a graph theoretic presentation of $\mathrm{C}_\mathrm{max}(\mathcal M)$ based on $G_\mathcal M$. First, we claim that:
\begin{equation}
    I(G_\mathcal M)=\{ S \subseteq \mathcal M \;|\;[P_i,P_j]=0,\;\forall P_i, P_j \in S\}\;,
\end{equation}

where we denoted $I(G_\mathcal M)$ as the set of independent subsets of $G_\mathcal M$. Graph theoretically, $S \in I(G_\mathcal M)$ iff its elements do not share an edge in $E(G_\mathcal M)$. For the frustration graph $G_\mathcal M$, this happens iff all the Paulis commute, since a pair of Paulis either commute or anti-commute, showing the algebraic presentation above. We say that $S \in I(G_\mathcal M)$ is also \emph{maximal} ($S \in I_\mathrm{max}(G_\mathcal M)$) if:
\begin{equation}
    \forall S^\prime\in I(G_\mathcal{M}),\; S\subseteq S^\prime\Rightarrow S^\prime =S \;.
\end{equation}

\textbf{Claim:} There is a correspondence $Q\in\mathrm{C}_\mathrm{max}(\mathcal M)\longleftrightarrow (S \in I_\mathrm{max}, f \in \mathrm B_S)$, for $S\in\mathrm{max}(G_\mathcal M)$ and $f\in \mathrm B_S$\footnote{Cf. Eq.~\eqref{admissiblef}.}.

In fact, consider $Q \in \mathrm{C}_\mathrm{max}(\mathcal M)$. Since every $P\in Q$ satifies either $P\in\mathcal{M}$ or $-P\in\mathcal{M}$, we can define $f_Q \in \{\pm\}^{Q}$ with elements:
    \begin{equation}
    (f_{Q})_P =  
    \begin{cases}
1 \; \mathrm{if\;} P \in \mathcal M\;,\\
-1 \; \mathrm{if\;}P \in -\mathcal M\;.
    \end{cases}\;\Rightarrow\;
    (\vec v_Q)_i =    
    \begin{cases}
(f_Q)_{P_i} \; \mathrm{if\;} \pm P_i \in Q\;,\\
0 \; \mathrm{otherwise}\;.
    \end{cases}
    \label{eq:Qvectorsigned}
    \end{equation} 
Thus, there is a subset $S_Q \subseteq \mathcal M$ such that we can write:
\begin{equation}
    Q =\{(f_Q)_P P\}_{P \in S_Q}\;.
\end{equation}
We note that:
\begin{itemize}
    \item $f_Q \in \mathrm B_{S_Q}$, since $-1 \notin \langle Q\rangle$ imposes that, for every $S^\prime \subseteq S_Q$:
    \begin{equation}
        \prod_{P \in S^\prime}(f_Q)_P P \neq -1\;.
    \end{equation}
    \item $S_Q \in I_\mathrm{max}(G_\mathcal M)$: We know that $S_Q \in I(G_{\mathcal M})$, since $Q$ is commuting. Furthermore, if there were $S^\prime\in I(G_{\mathcal M})$ with $S\subsetneq S^\prime$, we could add $P_i\in S^\prime\setminus S$ to $Q$, contradicting its maximality.
\end{itemize}
So, there is a map $Q \mapsto (S,f)$, for $S \in I_\mathrm{max}(G_\mathcal M)$ and $f \in \mathrm B_S$. We also denote $\vec v_{S_Q,f_Q} =\vec{v}_Q$.

We show that we can also reconstruct an element of $\mathrm C(\mathcal M)$ from a pair $(S \in I_\mathrm{max}, f \in \mathrm B_S)$. There are two possibilities:
\begin{itemize}
    \item $\mathrm B_S \neq \varnothing$: Then, note that, given $f \in \mathrm B_S$, we can construct:
    \begin{equation}
        Q_{(S,f)} \equiv \{f_P  P\}_{P \in S} \subseteq \pm \mathcal M\;.
    \end{equation}
    Since $f$ satisfies the Pauli product condition, we know that $-1 \notin \langle Q_{(S,f)}\rangle$. Furthermore, since $S$ is maximal, we know that:
    \begin{equation}
        \nexists P^\prime \in \mathcal M \backslash S \;:\;[P^\prime, S]=0\;,
    \end{equation}
    since otherwise $S \subsetneq S\cup\{P^\prime\} \in I(G_\mathcal M)$, contradicting the maximality of $S$. Note that this condition is robust to signs, since if two Paulis $P, P^\prime$ do commute, so does $[s_1 P, s_2P^\prime] = 0$ for all $s_1, s_2 \in \{\pm\}$. We then conclude that:
    \begin{equation}
        \nexists P^\prime \in \pm \mathcal M \backslash Q_{(S,f)} \;:\;[P^\prime, Q_{(S,f)}]=0\;.
    \end{equation}
    Then, $s \in \langle Q_{(S,f)}\rangle \cap \pm \mathcal M \backslash Q_{(S,f)} \subseteq \pm \mathcal M\backslash Q_{(S,f)}$ cannot be commuting with $Q_{(S,f)}$, and thus $Q_{(S,f)}\in\mathrm{C}_\mathrm{max}(\mathcal{M})$. We define on this case:
    \begin{equation}
        \vec{v}_{S,f}\equiv \vec{v}_{Q_{(S,f)}}
        \label{eq:Sfvectorsigned}
    \end{equation}.
    \item $\mathrm B_S =\varnothing$: Then there is no consistent stabilizer group among the signed Pauli elements of $S \in I_\mathrm{max}(G_\mathcal M)$. We call those cases trivial.
\end{itemize}
Hence, the set of non-trivial $(S \in I_\mathrm{max}(G_\mathcal M), f \in \mathrm B_S)$ are in bijection with $\mathrm C(\mathcal M)$. Eqs.~\eqref{eq:Qvectorsigned}, \eqref{eq:Sfvectorsigned} define $\vec{v}_{(S,f)}$ with:
\begin{equation}
(\vec v_{(S,f)})_i =    
\begin{cases}
f_{P_i} \; \mathrm{if\;} f_{P_i}P_i \in S\\
0 \; \mathrm{otherwise}\;.
\end{cases}\;,
\end{equation} 
and by Eq.~\eqref{eq:reducedstabcommuting} we have the main claim:
\begin{equation}
    \mathrm{STAB}(\mathcal M) =\mathrm{conv}\{\vec{v}_Q\}_{Q \in \mathrm C(\mathcal M)} =\mathrm{conv}\{\vec v_{S,f}\}_{S \in I_\mathrm{max}(G_\mathcal M)\;,\;f \in \mathrm B_S}\;,
\end{equation}
where both lists are of extremal vertices.

We now move to the size bound. The simplest upper bound one can take for the number of consistent $(S,f)$ pairs is by the following argument; any commuting set of $n$-qubit Paulis can be extended to a full stabilizer group, which is generated by some set of $n$ Paulis, $\{P_1,\cdots,P_n\}$. Then we are free to choose only $f_S(P_i)\in\{\pm1\}$, as the other Paulis are all products of some combination of the generators. This accounts for at most $2^n$ possible sign assignments of the extended stabilizer group. On the other hand, the number all sign assignments is $|\{\pm1\}^S|=2^{|S|}$. Then we can write:
\begin{equation}
    |\mathrm{stab}(\mathcal M)| =\sum_{S \in I_\mathrm{max}(G_\mathcal M)}|\mathrm B_S| \leq \sum_{S \in I_\mathrm{max}(G_\mathcal M)}2^{\min\{n,|S|\}}\leq 2^{\min\{n,s_{\mathrm{max}}\}}|I_\mathrm{max}(G_\mathcal M)| \;,
    \label{eq:sizebound}
\end{equation}
where 
\begin{equation}
    s_\mathrm{max} = \max_{S \in I_\mathrm{max}(G_\mathcal M)}|S|\;
    \label{eq:Ssize}
\end{equation}

For the term $|I_\mathrm{max}(G_\mathcal M)|$ we also have two separate upper bounds: first, for any graph of $m$ vertices, we know~\cite{moon1965cliques}:
\begin{equation}
    |I_\mathrm{max}(G_\mathcal M)|\leq 3^{m/3}\;.
\end{equation}
On the other hand, as implied by the arguments in the appendix $\ref{proof:Vrepalgo}$, every maximal commuting subset $S$ accounts for at least $2$ points in $\mathrm{stab}(\mathcal M)$. Then, it is bounded by the number of stabilizer groups, $2^{(1/2+o(1))n^2}$~\cite{aaronson2004improved}, accounting for the combined bound
\begin{equation}
    |I_\mathrm{max}(G_\mathcal M)|\leq 2^{\min\{cm,(1/2+o(1))n^2\}}\;.
    \label{eq:maxindepsetsize}
\end{equation}

and thus arriving at $|\mathrm{stab}(\mathcal M)| \leq 2^{\min\{n,m\}+\min\{cm,(1/2+o(1))n^2\}}$ with $c=\log 3/3$. For $m=\Omega(n^2)$, this just reduces to the loose bound $|\mathrm{stab}(\mathcal M)| \leq |\mathrm{stab}_n|$; however, for $m=O(n)$ we have $|\mathrm{stab}(\mathcal M)| = 2^{O(m)}$, dropping the explicit dependence on $n$.


Tightness of $\mathcal M = \{X_1,Y_1, Z_1, X_2, \cdots, Z_n\}$ follows by recognizing that the corresponding frustration graph is given as:
\begin{equation}
    G_\mathcal M =\underbrace{K_3 \sqcup K_3 \sqcup \cdots \sqcup K_3}_{m/3\mathrm{\;times}}\;,
\end{equation}
with $K_3$ being the completely connected graph of 3 vertices. Then, every independent set $S \in \mathrm{I}_\mathrm{max}(G_\mathcal M)$ can be written as: 
\begin{equation}
    S = \{S_1, S_2,\cdots, S_{m/3}\}\;,
\end{equation}
with $S_i \in \{X_i,Y_i, Z_i\}$ being one of the three independent sets of each $K_3$. Hence, we have indeed that $|I_\mathrm{max}(G_\mathcal M)|=3^{m/3}$. Furthermore, note that, for any $S^\prime \subseteq S$, their corresponding Pauli product is also independent, meaning that $\prod_{P \in S^\prime}P \neq \pm 1$, which implies that any $f \in \{\pm\}^{S}$ is also in $\mathrm B_\mathrm S$. Thus, $\mathrm B_S =\{\pm\}^{m/3}=\{\pm\}^{n}$ for any $S \in I_\mathrm{max}(G_\mathcal M)$, and we have that by Eq.~\eqref{eq:sizebound},
\begin{equation}
    |\mathrm{stab}(\mathcal M)| =\sum_{S \in I_\mathrm{max}(G_\mathcal M)}|\mathrm B_S| =|I_\mathrm{max}(G_\mathcal M)| 2^{n} = 2^{n+(\log 3/3)m}\;.
\end{equation}

\section{Proof of Theorem \ref{thm:Vrepalgo} \label{proof:Vrepalgo}}
We will present the proof by expanding the steps outlined in Alg. \ref{alg:Vrep}. In order to do the detailed analysis, we need to do some preliminaries on Pauli strings. A specific classical  representation of $P \in \mathcal P_n$ is in term of its \emph{symplectic representation}~\cite{gottesman1997stabilizer, gottesman1998heisenberg}, on which we write:
\begin{equation}
    P = \lambda X^\mathbf a Z^\mathbf b\;,
\end{equation}
where $\lambda \in \{\pm 1, \pm i\}$ and $\mathbf a, \mathbf b \in \mathbb F_2^n$, and we have used the shorthand notation:
\begin{equation}
X^\mathbf a\equiv \bigotimes_{i=1}^n X^{a_i} \quad, \quad Z^\mathbf b  \equiv \bigotimes_{i=1}^n Z^{b_i}\;.
\end{equation}
Considering $P_1 =\lambda_1 X^{\mathbf a_1} Z^{\mathbf b_1}$ and $P_2 =\lambda_2 X^{\mathbf a_2}Z^{\mathbf b_2}$, they satisfy:
\begin{equation}
    P_1 P_2 = (-1)^{\Omega_{1,2}} P_2 P_1\;,
\end{equation}
where $\Omega_{1,2}  = \mathbf a_1 \cdot \mathbf b_2+ \mathbf a_2 \cdot \mathbf b_1$ is referred as the symplectic form of the vectors $(\mathbf a_1,\mathbf b_1)^T \in \mathbb F_2^{2n}$ and $(\mathbf a_2,\mathbf b_2)^T \in \mathbb F_2^{2n}$.

Note that storing this bit representation of the Pauli, $P \mapsto \{\lambda , \mathbf a , \mathbf b\}$ requires $O(n)$ bits, and Pauli multiplication, obtained by considering $\{\lambda_1, \mathbf a_1, \mathbf b_1\}$ and $\{\lambda_2, \mathbf a_2, \mathbf b_2\}$, one can verify that $P_1 P_2$ is given by the representation $\{\lambda_1 \lambda_2(-1)^{\mathbf b_1\cdot \mathbf a_2}, \mathbf a_1 + \mathbf a_2, \mathbf b_1 + \mathbf b_2\}$, which requires $O(n)$ runtime, and so does checking commutativity of two Pauli strings, since it can be obtained by addition modulo 2 of the representation of the constituents. We will assume that every single Pauli is stored and manipulated using the symplectic representation. This includes the input, which is given as the following length $m$ array:
\begin{equation}
\mathcal M = \{P_1, P_2, \cdots, P_m\} \mapsto \{(\lambda_1, \mathbf a_1, \mathbf b_1),(\lambda_2, \mathbf a_2, \mathbf b_2), \cdots,  (\lambda_m, \mathbf a_m, \mathbf b_m)\}\;,
\end{equation}
stored in $O(mn)$ bits. We now go through the four steps:
\begin{enumerate}
    \item We compute the frustration graph datum by computing its adjacency matrix, which turns out to be the symplectic form itself, which, for $i,j \in [m]$, is given as:
    \begin{equation}
        \Omega_{i,j} \equiv \mathbf a_i\cdot \mathbf b_j+ \mathbf b_i \cdot\mathbf a_j \in \{0,1\}\;,   
    \end{equation}
    all of its elements can be computed on $O(m^2 \times n) = O(m^2 n)$ runtime;
    \item Equipped with a classical representation of the graph $G_\mathcal M$ by $(\mathcal M, [\Omega_{i,j}]_{i,j \in [m]})$, where $\mathcal M$ is the list of symplectic representations of the Paulis and the edges are specified by the adjacency matrix. It is known that there is an output-senstive algorithm~\cite{tsukiyama1977anew} that runs in $O(m^3 |I_\mathrm{max}(G_\mathcal M)|)$ runtime and $O(m^2)$ process space to compute $I_\mathrm{max}(G_\mathcal M)$;
     \item Given $S \in I_\mathrm{max}(G_\mathcal M)$, we will find every admissible sign assignment $ f \in \mathrm B_S$. Remember that $B_S$ is the set of all sign assignments that does not incur any product of the type $\prod_{P\in S^\prime}f_P P=-\mathds{1}$, for all subsets $S^\prime\subset S$.
\end{enumerate}
First, each $S\subseteq \mathcal M$ is specified as the list of symplectic representations, as:
    \begin{equation}
        S =\{(\lambda_{S,j}, \mathbf a_{S,j}, \mathbf b_{S,j})\}_{j =1}^{|S|}\;,
    \end{equation}
    referring to $(\lambda_{S,j}, \mathbf a_{S,j}, \mathbf b_{S,j})$ as the symplectic representation of $P_{S,j}$. We can construct the following matrix:
    \begin{equation}
        \mathbf M_S =
        \begin{pmatrix}
            \mathbf a_{S,1} & \mathbf a_{S,2} & \cdots & \mathbf a_{S, |S|}\\
            \mathbf b_{S,1} & \mathbf b_{S,2} & \cdots & \mathbf b_{S, |S|}
        \end{pmatrix} \in \mathbb F_2^{2n\times|S|}\;.
    \end{equation}
    Our algorithm is based on the following fact:
    
    \textbf{Claim: } There is a correspondence $\mathbf c \in \mathrm{ker}[\mathbf M_S] \subseteq \mathbb F_2^{|S|} \longleftrightarrow S^\prime \subseteq S\;:\;\prod_{P \in S^\prime }P=\pm\mathds{1}$.

    To see why, note that by definition, 
    \begin{equation}
        \mathbf c \in \mathrm{ker}[\mathbf M_S] \Leftrightarrow (\mathbf M_S)\cdot \mathbf c=\mathbf 0 \Leftrightarrow \sum_{i=1}^{|S|} \mathbf c_i \begin{pmatrix}
            \mathbf a_{S,i}\\
            \mathbf b_{S,i}
        \end{pmatrix} = \mathbf 0 \in \mathbb F_2^{2n} \;.
    \end{equation}
    Note that the lhs corresponds to the binary vector of the symplectic representation of $\prod_{j=1}^{|S|}P_i^{\mathbf c_j}$, and the rhs is the representation of the identity, modulo a $\{\pm1,\pm i\}$ phase. However, since a product of commuting Paulis can not square to $-\mathds{1}$, the phase is restricted to $\pm 1$. Thus:
    \begin{equation}
        \sum_{i=1}^{|S|} \mathbf c_i \begin{pmatrix}
            \mathbf a_{S,i}\\
            \mathbf b_{S,i}
        \end{pmatrix} = \mathbf 0 \Leftrightarrow \prod_{j=1}^{|S|}P_j^{\mathbf c_j}  =\prod_{P \in S^\prime}P =\pm\mathds{1}\;,
    \end{equation}
    where we defined $S^\prime \equiv \{P_j \in S\;|\;\mathbf c_j=1\}$. Consider now $S^\prime \subseteq S$ satisfying the constraint. Define the vector $\mathbf c \in \mathbb F_2^{|S|}$ by:
    \begin{equation}
        \mathbf c_j = 1[P_j \in S^\prime] \quad, \quad \forall j \in [|S|]\;,
    \end{equation}
    where the $j$th component is 1 if $P_j \in S^\prime$ and zero otherwise. Then:
    \begin{equation}
        \prod_{P \in S^\prime} P =\prod_{j=1}^{|S|}P_j^{\mathbf c_j} =\pm\mathds{1}\;.
    \end{equation}
    The binary vectors of the corresponding symplectic representation of both sides satisfies:
    \begin{equation}
        \sum_{i=1}^{|S|} \mathbf c_i \begin{pmatrix}
            \mathbf a_{S,i}\\
            \mathbf b_{S,i}
        \end{pmatrix} = \mathbf 0 \Leftrightarrow \mathbf c \in \mathrm{ker}[\mathbf M_S]\;,
    \end{equation}
    showing the claim.

    By performing Gaussian elimination, in $O(\min(n, |S|) n|S|)$ time~\cite{CLRS4}, we obtain a list of spanning vectors of its kernel:
    \begin{equation}
        \mathbf K=\begin{pmatrix}
            \mathbf c_1\\
            \mathbf c_2\\
            \cdots\\
            \mathbf c_D
        \end{pmatrix} \in \mathbb F^{D \times |S|}_2
        \label{eq:spankernel}
    \end{equation}
    where $D=\mathrm{dim} \mathrm{ker}[\mathbf M_S]$. 

    By the claim above, we can define a bitstring $\boldsymbol{\sigma} \in \mathbb F_2^{D}$ such that $\prod_{j=1}^{|S|}P_j^{[\bm{c}_i]_j} = (-1)^{\boldsymbol{\sigma}_i}\mathds{1}$ for all $i \in [D]$, where we denote $[\mathbf c_i]_j$ to refer to the $j$th component of the $i$th spanning vector of Eq.~\eqref{eq:spankernel}. A second claim is the following fact:
    
    \textbf{Claim: }Let $ f \in \{\pm\}^S$ be some sign assignment for each Pauli in $S$. Then, $f \in \mathrm B_S$ iff:
    \begin{equation}
        \mathbf K \cdot \mathbf x = \boldsymbol{\sigma}\;,
    \end{equation}
    where $f_i = (-1)^{x_i}$, defining $\mathbf x \in \mathbb F_2^{|S|}$. 
    
    To see why, note that membership in $\mathrm B_S$ is guaranteed if:
    \begin{equation}
        \prod_{P \in S^\prime} f_P P \neq -\mathds 1 \quad, \quad \forall S^\prime \subseteq S \quad \Leftrightarrow \quad \prod_{P \in S^\prime} f_P\mathds{1}  \neq -\prod_{P \in S^\prime} P \quad, \quad \forall S^\prime \subseteq S\;.
    \end{equation}
    Also note that we only need to check this condition for subsets $S^\prime$ such that $\prod_{P \in S^\prime} P \propto \mathds 1$, since it is automatically satisfied otherwise. By the correspondence between subsets with trivial product and the kernel of $\mathbf M_S$, this corresponds to checking:
    \begin{equation}
        \prod_{j=1}^{|S|} f_j^{\mathbf c_j}\mathds{1}  \neq -\prod_{j=1}^{|S|} P_j^{\mathbf c_j} \quad, \quad \forall \mathbf c \in \mathrm{ker}[\mathbf M_S] \quad \Leftrightarrow \quad \prod_{j=1}^{|S|} f_j^{[\mathbf c_i]_j}  \mathds{1}\neq -\prod_{j=1}^{|S|} P_j^{[\mathbf c_i]_j} \quad, \quad \forall i \in [D]\;,
    \end{equation}
    where we have used the fact that the rows of Eq.~\eqref{eq:spankernel} spans the kernel; for . In other words, the constraint must be:
    \begin{equation}
        \prod_{j=1}^{|S|}f_j^{[\mathbf c_i]_j} = \boldsymbol{\sigma}_i\quad, \quad \forall i \in [D]\;,
    \end{equation}
    since the lhs is restricted to $\{\pm1\}$. By the definitions of $\mathbf x, \boldsymbol{\sigma}$, the $D$ equations above corresponds to:
    \begin{equation}
        (-1)^{\sum_{j=1}^{|S|}[\mathbf c_i]_j x_j} = (-1)^{\sigma_i} \quad, \quad \forall i \in [D] \quad \Leftrightarrow \quad \mathbf K \cdot \mathbf x =\boldsymbol{\sigma}\;,
    \end{equation}
    as claimed.

    As a consequence, we can count the number of solutions. First, note that if the kernel $\mathrm{ker}[\mathbf M_S]$ is empty, i.e., there is no Pauli product proportional to the identity, then \emph{all} sign assignments $f\in\{\pm\}^S$ are admissible, accounting for $|B_S|=2^{|S|}$. If the kernel is not empty, then $K$ is full rank and consequently it is surjective. Given a particular solution $\mathbf x_p: \mathbf K \cdot\mathbf x_p=\boldsymbol{\sigma}$, the full set of solutions can be written as:
    \begin{equation}
        \mathrm B_S  \cong \mathbf x_p + \mathrm{ker}[\mathbf K] = \{\mathbf x \in \mathbb F_2^{|S|}\;:\; \mathbf x =\mathbf x_p+\mathbf k\;,\;\mathbf k \in \mathrm{ker}[\mathbf K]\}\;,
    \end{equation}
    where the isomorphism is between sets by the relation $\{f_j = (-1)^{x_j}\}_j$. Hence, $|\mathrm B_S| =2^{\mathrm{dim}\mathrm{ker}[\mathbf K]} = 2^{|S|-\mathrm{rank}[\mathbf K]}$, by the rank-nullity theorem of $\mathbf K \in \mathbb F_2^{D \times |S|}$. Since the row-space of $\mathbf K$ is $\mathrm{ker}[\mathbf M_S]$, we have that $|\mathrm B_S| = 2^{|S|-D}$. The rank-nullity theorem for $\mathbf M_S \in \mathbb F_2^{2n \times |S|}$ implies that $|S|-D=\mathrm{rank}[\mathbf M_S]$. We know that $\mathrm{rank}[\mathbf M_S] \leq |S|$, since $\mathbf M_S$ has $|S|$ columns. Furthermore, its column space forms a isotropic subspace of the symplectic vector space $(\mathbb F_2^{2n}, \Omega)$, defined as a subspace with vanishing symplectic form. It is known that such a subspace must have dimension $\leq n$~\cite{gottesman1997stabilizer}, yielding $\mathrm{rank}[\mathbf M_S] \leq n$. Combining the bounds, we get that:
    \begin{equation}
        |\mathrm B_S| =2^{|S|-D} =2^{\mathrm{rank}[\mathbf M_S]} \leq 2^{\min(n, |S|)}\;.
    \end{equation}
    Given $\mathbf K$ obtained through Gaussian elimination, we can now present the algorithm to append all the consistent $\{\vec{v}_{S,f}\}$: (1) Compute $\boldsymbol{\sigma}$ by evaluating the set $\{s_i = \prod_{j=1}^{|S|} P_j^{[\mathbf c_i]_j}\}_{i=1}^{D}$, that can be done in $O(D|S|n)  = O(|S|^2 n)$ time, (2) Perform Gaussian elimination on $\mathbf K$ to obtain the spanning set of vectors $\mathrm{ker}[\mathbf K] = \mathrm{span}_{\mathbb F_2} \{\mathbf c^\prime_1, \mathbf c^\prime_2, \cdots, \mathbf c^\prime_{|S|-D}\}$, that takes time $O(\min(D, |S|) D|S|) = O(|S|^3)$ and (3) for every $\mathbf c^\prime \in \mathrm{ker}[\mathbf K]$, compute and append the vector $\vec v_{S,f}$ in Eq.~\eqref{eq:consistentvector}, that takes $O(2^{|S|-D}\times m) = O(2^{\min(n,|S|)} m)$ time
    \footnote{On all runtime evaluations, we have used that $D \leq |S|$.}. Summing all up:
    \begin{align}
        &\underbrace{O(\min(n, |S|) n |S|) }_{\mathrm{Gaussian\;elimination\;to\;obtain\;}\mathbf K}+ \underbrace{O(|S|^2 n)}_{\mathrm{Computing\;the\;signs\;}\{s_i\}_i} +\underbrace{O(|S|^3)}_{\mathrm{Gaussian\;elimination\;in\;}\mathbf K} + \underbrace{O( 2^{\min(n, |S|)} m)}_{\mathrm{Computing \;the\;kernel\;and\;appending}}\\
        &=O(2^{\min(n, m)}m + m^3+m^2n + \min(m,n)nm)\;,
    \end{align}
    where on the last line we have used that $|S| \leq m$.


Combining all the three steps, we have our algorithm: (1) Compute the frustration graph, (2) list all the maximally independent subsets, and (3) construct all the vectors coming from consistent sign assignments of the independent subsets. Thus, we have, for the runtime:
\begin{align}
    T &= \underbrace{O(m^2 n)}_{\mathrm{computing\;} G_\mathcal M} + \underbrace{O(m^3|I_\mathrm{max}(G_\mathcal M)|)}_{\mathrm{computing\;independent\;subsets}} + \underbrace{|I_\mathrm{max}(G_\mathcal M)| \times O(2^{\min(n, m)}m + m^3+m^2n + \min(m,n)nm)}_{\mathrm{computing\;consistent\;{vectors\;for\;each\;} } S \in I_\mathrm{max}(G_\mathcal M)  }\;,
\end{align}
where we bounded $|S| \leq s_\mathrm{max}$, defined on Eq.~\eqref{eq:Ssize}. This is the full runtime depending on the frustration graph properties. To have an independent upper bound, we use that $s_\mathrm{max} \leq m$ and the bound in Eq.~\eqref{eq:maxindepsetsize}. Note that, for $n \to \infty$, and $m \leq \mathrm{poly}(n)$, the first term that builds the frustration graph is polynomial and can be dropped in the asymptotic runtime, since the other two are exponential. Furthermore, the last contribution has an extra exponential dependency coming from the sign assignments, making it the single dominant asymptotic term. We thus arrive at our claim:
\begin{equation}
    T = O(m2^{(\min(m,n) + \min(cm,(1/2+o(1))n^2)}+(m^3+m^2n+\min(m,n)nm)2^{\min(cm,(1/2+o(1))n^2)})\;,
\end{equation}
where $c=\log(3)/3$.

There are some interesting regimes: for a constant number of measurements, $m=O(1)$, the algorithm runs in linear time $T=O(n)$. If we allow logarithmic many measurements $m=O(\log{n})$, the time is still polynomial $T=O(\log(n)n^{1+\log 3/3})$. However, if we consider a quadratic (or larger) number of measurements $m=\Omega (n^2)$, then the bottleneck of the algorithm becomes the possibly exponentially many maximal commuting subsets, and the time scales with the size of $|\mathrm {stab}_n |$, giving $T=O(p(m,n)2^{O(n^2)})$.

$\square$
\section{Proof of Theorem \ref{thm:CMPtorSMPreduction} \label{proof:CMPtorSMPreduction}}
Let us first properly define the CMP as a decision problem:
\begin{defn}
    We define the classical marginal problem (CMP) on $n$ bits as the problem defined by:
    \begin{itemize}
        \item \textbf{Input: }A set of probability distributions $\{\vec{p}_{1}, \vec{p}_{2},\cdots, \vec{p}_{N}\}$, with $N \leq \mathrm{poly}(n)$, $\vec p_{i} \in \mathrm{\Delta}_{2^{|R_i|}}$, 
        where $R_i \subseteq[n]\;,\; |R_i| \leq k=O(1)$ for all $i \in [N]$;
        \item \textbf{Decision: }Output YES if exists a $n$ bit probability distribution $\vec{p} \in \Delta_{2^n}$ such that:
        \begin{equation}
            \sum_{\bar{\mathbf x} \in \mathbb F_2^{n-|R_i|}}p(\mathbf x, \bar{\mathbf x}) = p_i(\mathbf x)\;,\;\forall i \in [N]\;,
        \end{equation}
        and NO otherwise.
    \end{itemize}
    \label{def:classmargprob}
\end{defn}
We can now move to the two reductions:
\begin{enumerate}
    \item $\mathrm{CMP} \leq_\mathrm{poly} \mathrm{SMP}$: Given the inputs $\{\vec{p}_1, \vec{p}_2, \cdots, \vec{p}_N\}$, note that we can compute the following set of states
    \begin{equation}
        \left\{\phi_i\equiv \sum_{\mathbf x \in\mathbb F_2^{|R_i|}}p_i(\mathbf x )\ket{\mathbf x }\bra{\mathbf x} \right\}_{i=1}^N
    \end{equation}
    in $O(1)\times N \leq \mathrm{poly}(n)$ time. Since all the computational basis states are stabilizers, $\phi_i \in \mathrm{STAB}_{|R_i|}$, for all $i$, making it a valid input for the SMP. We claim that:
    \begin{equation}
        \exists \phi \in \mathrm{STAB}_n\;:\;\tr_{\bar{R}_i}\phi = \phi_i, \;\forall i \in [N] \Leftrightarrow \exists \vec{p} \in \Delta_{2^n}\;:\; \sum_{\bar{\mathbf x} \in \mathbb F_2^{n-|R_i|}}p(\mathbf x, \bar{\mathbf x}) =p_i(\mathbf x), \;\forall i \in [N]\;.
    \end{equation}
    \begin{itemize}
        \item ($\Leftarrow$): If there is a global probability distribution $\vec{p} \in \Delta_{2^n}$, note that we can construct:
        \begin{equation}
            \phi \equiv \sum_{\mathbf x \in \mathbb F_2^n } p(\mathbf x) \ket{\mathbf x}\bra{\mathbf x} \in \mathrm{STAB}_n\;,
        \end{equation}
        And then, given any $i \in [N]$:
        \begin{equation}
            \tr_{\bar R_i}\phi = \sum_{\mathbf x \in \mathbb F_2^{|R_i|}} \left( \sum_{\bar{\mathbf x} \in \mathbb F_2^{n-|R_i|}}p(\mathbf x, \bar{\mathbf x})\right)\ket{\mathbf x }\bra{\mathbf x} =\sum_{\mathbf x \in\mathbb F_2^{|R_i|}}p_i(\mathbf x )\ket{\mathbf x }\bra{\mathbf x} = \phi_i\;,
        \end{equation}
        since it correctly marginalizes to $\vec{p}_i$.
        \item ($\Rightarrow$) If there is a global stabilizer state $\phi \in \mathrm{STAB}_n$, consider the action of the completely dephasing channel on the computational basis, $\mathcal D(\cdot)$, that yields:
        \begin{equation}
            \mathcal D(\phi)=\sum_{\mathbf x \in \mathbb F_2^n}\bra{\mathbf x}\phi \ket{\mathbf x} \ket{\mathbf x }\bra{\mathbf x} \equiv \sum_{\mathbf x \in \mathbb F_2^n}p(\mathbf x) \ket{\mathbf x }\bra{\mathbf x}\;.
        \end{equation}
        We claim that $\vec p= (p(\mathbf x ) )_{\mathbf x \in \mathbb F_2^n}$ is our desired global probability distribution. From the commutativity of the dephasing channel with partial traces $\mathcal D[\tr_X(\cdot)] =\tr_X\mathcal D(\cdot)$, we have:
        \begin{equation}
            \mathcal D(\phi_i) = \tr_{\bar R_i}\mathcal D(\phi)\;,\;\forall i \in [N] \Leftrightarrow \sum_{\mathbf x\in \mathbb F_2^{|R_i|}} p_i(\mathbf x) \ket{\mathbf x}\bra{\mathbf x}=\sum_{\mathbf x \in \mathbb F_2^{|R_i|}} \left( \sum_{\bar{\mathbf x} \in \mathbb F_2^{n-|R_i|}}p(\mathbf x, \bar{\mathbf x})\right)\ket{\mathbf x }\bra{\mathbf x}\;,\;\forall i \in [N]\;,
        \end{equation}
        which is only true if $\vec{p}$ indeed is a YES instance of the CMP.
    \end{itemize}
    Hence, we have shown that under poly-time preprocessing, the decision of Def.~\ref{def:classmargprob} reduces to Def.~\ref{def:stabmargprob}.
    \item $\mathrm{SMP} \leq_\mathrm{poly} \mathrm{rSMP}$: Consider:
\begin{equation}
    \mathcal M_{1,\cdots, N} \equiv \bigcup_{i=1}^N \tilde{ \mathcal P}_{R_i}\;,
    \label{eq:tomographicdecomp}
\end{equation}
where $\tilde{\mathcal P}_{R_i}$ are all the non-trivial $4^{|R_i|}-1$ Pauli strings supported on $R_i$. Note that, furthermore:
\begin{equation}
    |\mathcal M_{1,\cdots,N}| \leq\sum_{i =1}^N(4^{|R_i|}-1) \leq N (4^k-1)
\end{equation}
is polynomially bounded in $n$, since $N \leq \mathrm{poly}(n)$ and $k=O(1)$. Since each $P$ is supported on some $R_i$, we can always compute its expectation value on some reduced density matrix, say $\phi_{i(P)}$. Then, we can define:
\begin{equation}
    \vec V_{1,\cdots, N} \equiv (\tr(P\phi_{i(P)}))_{P \in \mathcal M_{1, \cdots, N}}\;.
\end{equation}

Each expectation value can be computed in $O(1)$ time, and hence $\vec V_{1, \cdots,N}$ can be obtained in $\mathrm{poly}(n)$ time. We claim that $(\mathcal M_{1,\cdots, N}, \vec{V}_{1, \cdots, N})$ are the desired inputs for the reduced stabilizer membership problem: Let us show that
\begin{equation}
    \vec V_{1,\cdots, N} \in \mathrm{STAB}(\mathcal  M_{1,\cdots, N})\; \Leftrightarrow \; \exists \phi \in \mathrm{STAB}_n:\tr_{\bar{R}_i}\phi =\phi_i, \;\forall i \in [N]\;.
\end{equation}
\begin{itemize}
    \item ($\Leftarrow$): If there is a global stabilizer state that marginalizes correctly to the $\{\phi_i\}$, it follows that $\vec V_{1,\cdots, N} = \Pi_\mathcal M \vec{\phi} \in \mathrm{STAB}(\mathcal M)$;
    \item ($\Rightarrow$) If the vector lies in the reduced polytope, there is a $\phi \in \mathrm{STAB}_n$ such that $\vec V_{1, \cdots, N} =\Pi_\mathcal M \vec{\phi}$. However, note that, by definition on Eq.~\eqref{eq:tomographicdecomp}, $\mathcal  M_{1,\cdots, N}$ contains a tomographic complete set of Paulis on all $\{R_i\}$, which fixes all the $N$ reduced density matrices of the state $\phi$ to be exactly $\{\phi_i\}$. 
\end{itemize}
Again, we have shown that the decision of Defs. \ref{def:stabreducedprob} can be made with classical poly-time by considering Def. \ref{def:stabmargprob}. 
\end{enumerate}
$\square$

\end{document}